\def\format{A} 
\documentclass[11pt]{article}

\usepackage{ifthen}
\usepackage{soul}
\usepackage{fullpage, prettyref}
\usepackage{fancyhdr}
\usepackage{graphicx}
\usepackage{marvosym}
\usepackage{ifthen}
\usepackage{algorithm}

\usepackage{mdframed}
\usepackage{amsmath, amssymb, amsthm}
\usepackage{url}
\usepackage{fullpage, prettyref}
\usepackage{boxedminipage}
\usepackage{wrapfig}
\usepackage{ifthen}
\usepackage{color,xcolor}
\usepackage{transparent}
\usepackage{enumitem}
\usepackage{varwidth}
\usepackage[pagebackref,colorlinks=true,pdfpagemode=none,urlcolor=blue,linkcolor=blue,citecolor=violet,pdfstartview=FitH]{hyperref}
\usepackage[noabbrev,nameinlink]{cleveref}

\def\2plus{{\tt (++)}}
\def\3plus{{\tt (+++)}}
\def\4plus{{\tt (++++)}}
\def\5plus{{\tt (+++++)}}

\usepackage{ifthen}
\usepackage{soul}
\usepackage{fullpage, prettyref}
\usepackage{fancyhdr}
\usepackage{color,xcolor}

\usepackage{times}
\usepackage{graphicx}
\usepackage{marvosym}
\usepackage{ifthen}
\usepackage{algorithm}

\usepackage[noend]{algpseudocode}
\algrenewcomment[1]{\(\triangleright\) {\small{\emph{\color{blue} #1}}}}
\algnewcommand{\LineComment}[1]{\State \(\triangleright\) \emph{\color{blue} #1}}
\algnewcommand{\Invariant}[1]{\State \(\triangleright\) \emph{\color{red} #1}}

\usepackage{amsmath, amssymb, amsthm}
\usepackage{url}
\usepackage{fullpage, prettyref}
\usepackage{boxedminipage}
\usepackage{wrapfig}
\usepackage{ifthen}

\usepackage{transparent}
\usepackage{enumitem}
\usepackage{varwidth}

\theoremstyle{definition}
\newtheorem{theorem}{Theorem}
\newtheorem{lemma}{Lemma}
\newtheorem{claim}{Claim}

\newtheorem{observation}{Observation}

\newcommand{\ignore}[1]{}

\def\2plus{{\tt (++)}}
\def\3plus{{\tt (+++)}}
\def\4plus{{\tt (++++)}}
\def\5plus{{\tt (+++++)}}

{\hspace*{\fill}$\Box$\par}

\newlength{\algobox}
\newcommand{\OPT}{\mathtt{OPT}}
\newcommand{\ALG}{\mathtt{ALG}}
\newcommand{\opt}{\mathtt{opt}}

\newcommand{\calA}{\mathcal{A}}

\newcommand{\calD}{\mathcal{D}}

\newcommand{\calI}{\mathcal{I}}

\newcommand{\calM}{\mathcal{M}}

\newcommand{\calP}{\mathcal{P}}

\newcommand{\RR}{\mathbb{R}}

\DeclareMathOperator*{\Exp}{\ensuremath{{\mathbf{Exp}}}}

\DeclareMathOperator*{\Prob}{\ensuremath{\mathbf{Pr}}}
\renewcommand{\Pr}{\Prob}

\newcommand{\eps}{\ensuremath{\varepsilon}}

\newcommand{\grad}{\nabla}

\newcommand{\etal}{{\it et al.\,}}

\newcommand{\eq}{\leftarrow}

\def\primdual{\mathsf{PrimDual}}
\def\greedydual{\mathsf{GreeDual}}
\def\bqs{\mathsf{BQS}}

\newcommand{\luc}[1]{{\color{red} {\bf luc:} #1}}

\title{A Primal-Dual Analysis of Monotone Submodular Maximization}
\author{Deeparnab Chakrabarty \and Luc Cote}
\date{}
\newcommand{\ipco}[1]{}
\newcommand{\arxiv}[1]{#1}
\usepackage{accents}
\usepackage{bbm}
\newcommand*{\dt}[1]{%
	\accentset{\mbox{\large\bfseries .}}{#1}}

\def\lp{\mathsf{lp}}
\def\dual{\mathsf{dual}}

\newcommand{\msm}{$k$-MSM}
\newcommand{\mmsm}{matroid-MSM}
\newcommand{\spann}{\mathsf{span}}
\DeclareMathOperator*{\argmax}{arg\,max}
\begin{document}

\maketitle

\begin{abstract}
In this paper we design a new primal-dual algorithm for the classic discrete optimization problem of maximizing a monotone submodular function subject to a cardinality constraint achieving the optimal approximation of $(1-1/e)$. This problem and its special case, the maximum $k$-coverage problem, have a wide range of applications in various fields including operations research, machine learning, and economics. While greedy algorithms have been known to achieve this approximation factor, our algorithms also provide a dual certificate which upper bounds the optimum value of any instance. This certificate may be used in practice to certify much stronger guarantees than the worst-case $(1-1/e)$ approximation factor.
\arxiv{\newline\indent}
We additionally consider the generalized problem of maximizing a monotone submodular function subject to a matroid constraint. Here we discuss an analysis of the greedy and continuous greedy and for each algorithm present a dual-fitting technique to achieve an instance-specific upper-bound certificate.
\end{abstract}
\setcounter{page}{1}
\section{Introduction}
A set function $f: 2^V \to \RR$ defined over subsets of a ground set $V$ of $n$ elements is submodular if for all subsets $A$ and $B$, $f(A) + f(B)\geq f(A\cap B) + f(A\cup B)$. Such a function is monotone if $f(A)\leq f(B)$ whenever $A\subseteq B$. The monotone submodular maximization problem under a cardinality constraint takes an evaluation oracle to a monotone submodular function $f$, a positive integer $k$, and wishes to find a subset $S\subseteq V$ with $|S| = k$ which maximizes $f(S)$. We call this the \msm~problem in short.

Submodular functions arise in many areas, and therefore, \msm~is a fundamental optimization problem with a wide range of applications in various fields including facility location~\cite{CornFN77,NemhW78}, viral marketing~\cite{KempKT03}, sensor placement~\cite{KrauG05}, document summarization~\cite{LinB11}, web searching and indexing~\cite{DasgGKOPT07,BaezBC15} and algorithmic economics~\cite{PapaSS08,DughR09,SchuU13}. 
An important special case of submodular functions are {\em coverage functions} where the ground set itself corresponds to sets of a set-system, and $f(S)$ indicates the size/weight of the union of the sets indexed by $S$. The \msm~problem for coverage functions is called the maximum $k$-coverage problem, and many of the applications~\cite{KrauG05,LinB11,BaezBC15} mentioned above are in fact special cases of this problem. 
More than forty years ago, Nemhauser, Wolsey and Fisher~\cite{NemhWF78} proved that a simple greedy algorithm for \msm, namely 
one that proceeds in $k$ rounds, and in each round picking an element which leads to the largest function value increase, achieves an $\eta_k := 1 - \left(1-1/k\right)^k \to (1 - 1/e)$-factor approximation to the optimum value. More precisely, for any monotone submodular function $f$, if $A$ is the output of the above algorithm, then $f(A) \geq \eta_k \cdot \max_{S \subseteq V:|S|=k} f(S)$.
Complementing the above result, Nemhauser and Wolsey~\cite{NemhW78} showed that no algorithm making polynomially many evaluation queries to a submodular function can achieve an approximation factor exceeding $(1-1/e)$. 
For the special case of coverage functions, Feige~\cite{Feige98} proved that it is NP-hard to obtain an approximation factor better than $(1-1/e)$.

Although the worst-case approximation factor of any algorithm for \msm~cannot be better than $(1-1/e)$, it has been noted in many works~\cite{KrauSG08,BaezBC15,SakaI18,SomaY18,PokuST20,BalkQS21} that the greedy algorithm mentioned above often performs much better on typical instances that arise in practice. Note that to make such an assertion empirically, one needs to either compute the optimum value or an upper bound on the optimum value of an \msm~instance. For instances with very small size, one could imagine a brute force search. 
As early as 1981, Nemhauser and Wolsey~\cite{NemhW81} proposed a non-trivial integer linear program with exponentially many constraints that exactly captures \msm, and then suggested constraint generation and branch-and-bound methods to solve~\msm~exactly. Indeed, this is what Krause, Singh, and Guestrin~\cite{KrauSG08} do, but this integer programming approach also only works for small instances 
(\cite{KrauSG08} report $n=16$ and $k\leq 5$). 
For succinct submodular functions such as coverage functions, \msm~can be expressed succinctly as a linear program and this was used by Baeza-Yates, Boldi, and Chierichetti~\cite{BaezBC15} as an upper bound to compare an algorithm's performance. For general monotone submodular functions, other works (eg, \cite{SakaI18,BalkQS21} and more examples within) have proved upper bounds on the optimum value which can be computed fast and used that as the benchmark. However, the {\em worst-case} guarantees of such upper bounds were either missing (in case of Sakaue and Ishihata's work~\cite{SakaI18}) or sub-optimal (Balkanski, Qian, and Singer~\cite{BalkQS21} could only prove a worst-case $1/2$-approximation for their novel upper-bound). At this point, we remark that the BQS-upper bound~\cite{BalkQS21} is empirically the lowest (that is, the best) upper bound.
We were motivated by the question: 
\emph{is there a principled and efficient way to obtain upper bounds on \msm~which have optimal worst case guarantees and also have good empirical performance? }\smallskip

Our answer is in the affirmative. To do so, we consider the {\em dual} of the Nemhauser-Wolsey integer linear program alluded to above. Our main contribution is a primal-dual algorithm (see~\Cref{sec:2}) with respect to this linear program which, for any monotone submodular function $f$, returns a set $S$ with $k$ elements, and a feasible dual solution of value $\dual \geq \opt$ with guarantee that $f(S) \geq \eta_k\cdot \dual$.
One benefit of a primal-dual algorithm is that it immediately provides instance-wise guarantees (by simply comparing $f(S)$ and $\dual$), and since the ``dual portion'' of our  algorithm is designed to try and decrease the dual value as much as possible, the typical ratio that we empirically observe (see~\Cref{sec:experiment}) is much higher than $(1-1/e)$ and comparable to the current state-of-the-art empirical upper bounds due to Balkanski, Qian, and Singer~\cite{BalkQS21}. We remark that 
our ``primal'' algorithm is different from the greedy algorithm, and there exist instances where it out-performs the greedy algorithm. Having said that, we must also say that in almost all our empirical instances the greedy algorithm usually does as well or (slightly) better. Furthermore, this principled way of obtaining instance-wise upper-bounds readily suggests generalizations to richer constraints than just the cardinality constraint. For instance, for matroid constraints we can provide (see~\Cref{sec:matroid-dualfit}) dual-fitting explanations for the greedy and continuous-greedy algorithms opening up possibilities of empirical evaluations to see if these algorithms typically do give much better than their worst-case approximation factors even in these richer constraint situations.


At a technical level, from an approximation algorithms point of view, we believe that our primal-dual algorithm for a {\em maximization} problem is of interest.
A primal-dual algorithm for an optimization problem which can be expressed as an integer linear program {\em simultaneously} 
constructs a feasible (integer) primal solution and a feasible dual solution. In this schema, usually the process proceeds in iterations with the primal and dual steps informing each other in some fashion.
The first primal-dual algorithms were designed for {\em exact} algorithms with Kuhn's Hungarian method~\cite{Kuhn55} to solve the maximum weight bipartite matching being the classic 
example. 
The first primal-dual {\em approximation} algorithm is due to Bar-Yehuda and Even~\cite{BarYE81} who described $2$-approximate primal-dual algorithm for the vertex cover problem.
Primal-dual algorithms have had a lot of success (see~\cite{GoemW97,Will02,BarYR05} for surveys) for designing approximation algorithms for {\em minimization} problems like~\cite{AgraKK91,GoemW95,BafnBF99,WillGMV95,JainV01} to name a few.
On the other hand, we are not aware of many works which give primal-dual approximation algorithms for maximization problems. 
Indeed, this was stated as an ``open problem'' in the late 90s, and the first such result is due to Bar-Noy, Bar-Yehuda, Freund, Naor, and Schieber~\cite{BarBFNS01} who gave a general framework for solving many allocation and scheduling problems to maximize throughput/utility. 
Their algorithms are based on the local-ratio framework~\cite{BarYE81,BarYR05}, and can be interpreted as a primal-dual algorithm which starts with an infeasible dual solution which proceeds in rounds picking violated dual constraints and picking corresponding primal variables, raising dual variables to fix this constraint, and terminating once one reaches dual-feasibility. 
In contrast, our algorithm always maintains a feasible dual solution and proceeds in rounds {\em decreasing} the dual objective, making primal picks when we can't, and terminating when we can't make any more primal choices. 
One other primal-dual approximation algorithm for a maximization problem that we are aware of is for the (offline) Adwords problem (also called the maximum budgeted allocation problem); Chakrabarty and Goel~\cite{ChakG10} describe a primal-dual $\frac{3}{4}$-approximation for this problem, and their algorithm 
also falls in our class which maintains a feasible dual solution throughout. 
We think more research is warranted for primal-dual algorithms for NP-hard maximization problems.


We also remark that the primal-dual technique also has had a tremendous influence in the design of online algorithms following the seminal works of Buchbinder and Naor~\cite{BuchN09,BuchN2009}. 
In particular, for problems such as online matching and its generalization the online Adwords problem, this paradigm can give optimal $(1-1/e)$-competitive algorithms for the maximum matching~\cite{DevaJK13} and the Adwords problem~\cite{BuchJN07} (with small bid-to-budget ratio), reproving earlier results of~\cite{KarpVV90} and~\cite{MSVV07}. Primal-dual techniques have also been applied in the streaming setting for matching problems, both with linear~\cite{PazS18,GhafW18} and submodular~\cite{LeviW21} objectives.
All these algorithms can indeed be reinterpreted as primal-dual approximation algorithms, but they either have exact or much better approximation ratios in the offline setting.

\subsection{Related Works, Other Results, and Future Directions}
There is a plethora of works related to the \msm~problem, and we won't attempt to survey these (see~\cite{KraG14,BuchF18} and citations within for such surveys) in this paper. Rather, we mention a few relevant works most related to our results. Wolsey~\cite{Wolsey82} studied the submodular set cover problem where each element $j\in V$ has a cost and the objective is to find the minimum cost subset $S$ with $f(S) = f(V)$. This problem generalizes the minimum set cover problem.~\cite{Wolsey82} analyzed a natural greedy algorithm via a dual-fitting approach to a minimization linear program with exponentially many constraints akin to the Nemhauser-Wolsey LP~\cite{NemhW81}. Later on, Fujito~\cite{Fujito99} described a primal-dual algorithm based on this LP for the submodular set cover problem,
which is a {\em minimization} problem, and in particular, Fujito's algorithm generalized the primal-dual algorithm for the set cover problem. Note, there is an $f$-approximate primal-dual algorithm where $f$ is the maximum number of sets an element can appear in. Fujito's approximation factor is some function of the submodular function which precisely equals $f$ when the submodular function is the coverage function. 

Recently, a primal-dual algorithm was designed for a submodular matching problem by Levin and Wajc~\cite{LeviW21} in the streaming setting. In particular, they consider the problem of edges of a bipartite graph streaming by, and the objective is to pick a matching with as large a submodular function value as possible. In doing so,~\cite{LeviW21} use the concave closure of the submodular function, and use this to generalize a primal-dual based $(2+\eps)$-approximation streaming algorithm for weighted matching due to Paz and Schwartzman~\cite{PazS18}. Their algorithm is quite different and rather unrelated from the primal-dual algorithm we design for \msm.

Submodular function maximization has also been extensively studied under richer constraints. For instance, when one constrains the returned set $S$ to be an independent set in a matroid, one obtains the \mmsm~problem. Fisher, Nemhauser, and Wolsey~\cite{FishNW78} proved that a natural greedy algorithm obtains an $\frac{1}{2}$-approximation, but improving this remained an open problem till the seminal work of Calinescu, Chekuri, P\'al, and Vondr\'ak~\cite{CaliCPV07,CaliCPV11} which studied continuous extensions of submodular functions. Indeed, in their IPCO version~\cite{CaliCPV07}, the authors describe an upper bound on $\opt$ which is precisely the Nemhauser-Wolsey LP value (when properly generalized to capture matroid constraints; see~\Cref{sec:2}), and they prove that this LP has an integrality gap of $(1-1/e)$. 
As we mention earlier, we can obtain ``dual-fitting'' results for both the deterministic and the continuous greedy algorithms. More precisely, one can design algorithms which construct dual solutions to the (generalized) Nemhauser-Wolsey LP and compare the ``primal'' solutions of these algorithms to these duals (see~\Cref{sec:matroid-dualfit}). 
We stress that such a result is probably implicit in earlier works~\cite{CaliCPV11,ChekVZ18}, but we believe making this explicit is worthwhile. In particular, it could lead to ideas for obtaining {\em deterministic} primal-dual algorithms for \mmsm~which obtain an $(1-1/e)$-approximation. This is an open question; currently the only other $(1-1/e)$ algorithm via local search due to Filmus and Ward~\cite{FilmW14} is also randomized, and the best deterministic approximation factor is $0.5008$ due to Buchbinder, Feldman and Garg~\cite{BuchFG19}.

Another generalization which may be worth studying with the primal-dual lens is when the constraint is a knapsack constraint. When the submodular function is a coverage function, the problem is the budgeted coverage problem for which Khuller, Moss, and Naor~\cite{KhulMN99} gave the first $(1-1/e)$-approximation. This was generalized to general monotone submodular functions by Sviridenko~\cite{Svi04}. We leave the study of these problems using the dual lens to obtain better instance-wise approximations as a future direction of study.

\subsection{Technical Preliminaries}
Fix a monotone submodular function $f:2^V\to \RR$ defined over subsets of a ground set $V$. We will assume $f(\emptyset) = 0$ which implies the range of $f$ is non-negative due to monotonicity.
For any subset $S\subseteq V$ and $j\in V$, we use the notation $f_{S}(j) := f(S\cup j) - f(S)$ to denote the marginal utility of adding $j$ to $S$. Note that monotonicity implies that $f_S(j) \geq 0$ for all $j\in V$, and note that $f_S(j) = 0$ for all $j\in S$. An equivalent definition of submodularity is
\[
\forall S\subseteq T, \forall j\in V~~~~ f_S(j) \geq f_T(j)
\]
The following fact is key for the Nemhauser-Wolsey~\cite{NemhW81} formulation (we provide the short proof for completeness.).
\begin{claim}[Proposition 2.2 in~\cite{NemhWF78}]\label{clm:fnw}
	For a non-negative, monotone, submodular function $f:2^V \to \RR_+$ and for any two subsets $S, T$ of $V$, $f(S	) + \sum_{j\in T} f_S(j) \geq f(T)$.
\end{claim}
\arxiv{
\begin{proof}
	Order the elements $\{j_1, \ldots, j_r\}$ of $T\setminus S$ arbitrarily. Now, note that 
		    $f(T\cup S) - f(S)$ equals $\sum_{i=1}^r f_{S \cup \{j_1, \ldots, j_{i-1}\}}(j_i)$
		    which, by submodularity, is $\leq\sum_{i=1}^r f_S(j_i) = \sum_{j\in T\setminus S} f_S(j) = \sum_{j\in T} f_S(j)$.
%
	Rearranging, we get $f(S) + \sum_{j\in T} f_S(j) \geq f(T\cup S)$ which is at least $f(T)$ by monotonicity.	
\end{proof}
}
\noindent
We use the following analytic fact that forms the basis for almost every $\eta_k$-factor analysis.
\begin{claim}\label{clm:ana}
	Let $0 = x^{(1)}, x^{(2)}, \ldots, x^{(k+1)}$ be a sequence of numbers satisfying
	\[
	\forall 1\leq i\leq k,~~~ x^{(i+1)} \geq \left(1 - \frac{1}{k}\right)x^{(i)} + \frac{1}{k}\cdot z
	\]
	for some number $z$. Then, $x^{(k+1)} \geq \eta_k \cdot z$ where $\eta_k = 1 - \left(1 - \frac{1}{k}\right)^k$
\end{claim}
\begin{proof}
	Using induction, one easily notices that $x^{(i)} \geq \left(1 - \left(1 - \frac{1}{k}\right)^{i-1}\right)z$ for all $1\leq i\leq k+1$.
\end{proof}

\section{A Primal-Dual Algorithm under the Cardinality Constraint}\label{sec:2}

The following two linear programs are the primal formulation of \msm~and its dual.\\
\begin{minipage}{0.46\textwidth}
\ipco{\fontsize{8pt}{12pt}}
	\begin{alignat}{6}
		\lp:=&&\max~~~~\lambda~&&~&\tag{\sf{primal}} \label{eq:primal} \\
		&&\sum_{j \in V} x_j &~=~ && k~&\tag{P1}\label{eq:p1}\\
		&&\lambda - \sum_{j \in V} f_{S}(j)x_j &~\leq~ && f(S), && ~\forall S \subseteq V \tag{P2}\label{eq:p2}\\
		&&x_j &~\geq~&& 0, && ~\forall j\in V \notag
	\end{alignat}
\end{minipage} ~~~\vline
\begin{minipage}{0.5\textwidth}
\ipco{\fontsize{8pt}{12pt}}
	\begin{alignat}{6}
		\min~   &&~k\alpha + \sum_{S\subseteq V}\tau_S f(S)  &=:&& \dual \notag\\
		&&\alpha - \sum_{S\subseteq V}\tau_S f_{S}(j) &~\geq~ && 0,~ \forall j \in V \tag{D1} \label{eq:d1}\\
		&& \sum_{S\subseteq V} \tau_S &~=~&& 1 \tag{D2} \label{eq:d2}\\
		&& \tau_S &~\geq~&& 0, ~\forall S\subseteq V \tag{D3} \label{eq:d3}
	\end{alignat}
\end{minipage}
\vspace{1ex}
\paragraph{The Primal.} As mentioned in the introduction, the primal LP was first explicitly studied by Nemhauser and Wolsey~\cite{NemhW81}
who used~\Cref{clm:fnw} to prove that when $x_j \in \{0,1\}$, the value of the integer program is precisely the optimum value\arxiv{\footnote{To quickly see why $\lp \geq \opt$,
	let $S^*$ be the optimal solution and consider $x_j = 1$ for $j\in S^*$ and $x_j = 0$ for $j\in V\setminus S^*$ with  $\lambda = f(S^*)$.
	$(x,\lambda)$ is a feasible solution because for any subset $S\subseteq V$ we get
	$f(S) + \sum_{j \in V} f_{S}(j)x_j ~=~ f(S) + \sum_{j\in S^*} f_S(j)$ and the latter is $\geq f(S^*) = \lambda$ by~\Cref{clm:fnw}.}} of the \msm~problem.
The motivation of~\cite{NemhW81} was to use constraint generation and branch-and-bound methods to exactly solve a \msm~instance. 

\sloppy Calinescu \etal~\cite{CaliCPV07} also considered the above primal as an extension of a submodular function. In particular, given any $\vec{x}\in [0,1]^n$, 
they define $f^*(\vec{x}) := \min_{S\subseteq V} \left(f(S)  +  \sum_{j\in V} f_S(j)x_j\right)$. Note that the primal $\lp$ is precisely $\max_{\vec{x}: \sum_{j\in V} x_j = k} f^*(\vec{x})$.
Indeed, the main motivation of~\cite{CaliCPV07} was when $\vec{x}$ lay in the independent set polytope of a matroid, and they describe a randomized algorithm\footnote{For the cardinality constrained case this is much simpler and described in~\Cref{sec:notes-on-primal}.} which takes such a point $\vec{x}$ and produces an independent subset $R$ with $\Exp[f(R)] \geq (1-1/e)\cdot f^*(\vec{x})$. This proves that the integrality gap of the LP is at most $(1-1/e)$. Furthermore, it can be shown\footnote{in their analysis, one can replace $\OPT$, the value of the optimal integral solution, by $f(\vec{x^*})$ and get this result; we leave these details out from this version which is restricted to the cardinality case} that the randomized solution $R$ returned by continuous-greedy algorithm from their journal version~\cite{CaliCPV11} 
also satisfies $\Exp[f(R)] \geq (1-1/e)\cdot \lp$.

On the other hand, the primal LP is not a computationally amenable object. Calinescu~\etal~\cite{CaliCPV07} (Theorem 3) assert that {\em evaluating} $f^*(\vec{x})$ is an NP-hard problem even when the submodular function is a coverage-type function. In~\Cref{sec:notes-on-primal}, we describe a submodular function $f$ and a vector $\vec{x}$
where it takes exponentially many calls to the submodular function to evaluate $f^*(\vec{x})$. In sum, we do not expect efficient algorithms to compute $\lp$ and thus cannot use it to compute a  certified upper bound on the optimum value for a given instance.

\paragraph{The Dual.} 
The dual program has a variable $\alpha$ corresponding to \eqref{eq:p1} and a variable $\tau_S$ for every subset $S\subseteq V$ corresponding to \eqref{eq:p2}. Note that by weak duality {\em any} 
feasible dual solution provides an upper bound on the optimum value. 
We desire fast algorithms which return a dual feasible solution with as small a value to get good upper bounds on the optimum value.
In the next section we describe a primal-dual algorithm which returns a feasible pair of solutions to both linear programs
with the guarantee that their values are within $(1-1/e)$. 

Note that \eqref{eq:d2} and \eqref{eq:d3} imply that $\tau$'s form a probability distribution over the subsets of $V$, and we abuse notation $S\sim \tau$ to denote a draw of a subset from this distribution. Using this, we note that $\sum_{S \subseteq V} \tau_S f(S) = \Exp_{S\sim \tau} [f(S)]$, the expected value of the submodular function wrt $\tau$. We use $\gamma^{(\tau)}$ to denote this, and henceforth lose the superscript due to brevity.
For any element $j\in V$, we note that $\sum_{S\subseteq V} \tau_S f_S(j) = \Exp_{S\sim \tau} [f_S(j)]$ is the expected marginal value of $j$ wrt $\tau$. We use $\beta^{(\tau)}_j$ to denote this, 
and henceforth lose the superscript due to brevity. Now note that \eqref{eq:d1} implies that once $\tau$ is fixed, the parameter $\alpha = \max_{j\in V} \beta_j$, that is, $\alpha$ is the maximum expected marginal of an element. Thus, once we fix the distribution $\tau$, the dual objective $\dual(\tau)$ is the sum of the expected value of $f(S)$ plus $k$ times the {\em maximum} expected marginal of an element.
Therefore, we can succinctly write the dual LP as
\begin{equation}\label{eq:dual-reinter}
\ipco{\small}
	\dual(\tau) =   k~\cdot \underbrace{\max_{j\in V} ~~~\underbrace{\left(\Exp_{S\sim \tau} [f_S(j)]  \right)}_{=\beta_j}}_{=\alpha} ~~~~+~~~ \underbrace{\Exp_{S\sim \tau} [f(S)]}_{=\gamma}    ;~~~~~~~ \dual = \min_{\tau~:~\text{distr. on } S\subseteq V} ~\dual(\tau)
\end{equation}

\arxiv{To repeat, for every distribution $\tau$ the value $\dual(\tau)$ is an upper bound on the optimum value. A trivial example is the ``point-distribution'' which puts all its mass on the empty set, that is, $\tau_\emptyset = 1$
and $\tau_S = 0$ for all $S\neq \emptyset$. Then, the dual objective becomes $k\cdot \max_{j} f(\{j\})$ which is a trivial upper bound on $\opt$. Next consider the greedy algorithm's solution $S$ and consider the dual solution putting it all its mass on $S$, that is, $\tau_S = 1$. Using submodularity and the greedy property, one readily observed that $f(S) \geq k\cdot \max_j f_S(j)$, and thus, the dual objective is at most $2f(S)$ proving that the greedy is at least $1/2$-approximate.
It is, in fact, not too difficult 
to construct 
a distribution $\tau_{\mathsf{Greed}}$ such that the value of the greedy solution is {\em at least} $(1-1/e)\cdot \dual(\tau_{\mathsf{Greed}})$. We include a discussion of such ``dual fitting'' results in~\Cref{sec:dual-fitting}.}

%
%

\subsection{Primal-Dual Algorithm}

In this section we describe a simple deterministic algorithm which simultaneously returns a set $\ALG$ with $k$ elements, and a distribution $\tau$ over sets such that $f(\ALG)\geq (1-1/e)\cdot \dual(\tau)$. 
We begin with $\tau$ being the point distribution with all it mass on the empty set (and thus $\dual(\tau)$ is simply $k$ times the largest singleton value). The algorithm proceeds in $k$-rounds and each round has two phases. In the first phase, we try to modify the distribution $\tau$ to {\em decrease} the dual objective. This will be done in a specific fashion which we describe shortly.
If we are unable to decrease the dual objective, we move to the second phase of making a primal pick:
 we select an element with the largest expected {\em marginal} increase, that is, the largest $\beta_j = \Exp_{S\sim \tau} [f_S(j)]$ into our solution and move to the next round. A feature of our algorithm is that we continually decrease the dual objective value, and this probably explains why we obtain good dual solutions empirically.

Before describing our primal-dual algorithm, it may help to interpret the greedy algorithm in the above framework where the distributions $\tau$ change discretely.
First note that in the beginning when $\tau$ concentrates all its mass on the empty set, 
the $\beta_j$ for every element $j$ is precisely $f(\{j\})$. At this point, if we make a primal pick with respect to the largest $\beta_j$, then we would be indeed taking the first step which the greedy algorithm takes. For simplicity, let this element be called $e_1$. Now, suppose we completely transfer the mass of $\tau$ from the empty set to the set $\{e_1\}$. Note that $\beta_j$ now is precisely $f(\{j,e_1\}) - f(\{e_1\})$, and so the primal pick, say $e_2$, with respect to the largest $\beta_j$ would indeed be the next greedy step. However, consider what happens to $\dual(\tau)$ as one moves from $\tau^{(0)}$ which has all its mass on $\emptyset$ to the distribution $\tau^{(1)}$ which puts all the mass on $\{e_1\}$. Note that 
\arxiv{\[
\dual(\tau^{(0)}) = k\cdot f(\{e_1\})~~~\text{and}~~~\dual(\tau^{(1)}) = k\cdot (f(\{e_1,e_2\}) - f(\{e_1\})) + f(\{e_1\})
\]}
\ipco{$\dual(\tau^{(0)}) = k\cdot f(\{e_1\})$ and $\dual(\tau^{(1)}) = k\cdot (f(\{e_1,e_2\}) - f(\{e_1\})) + f(\{e_1\})$}
and it is not necessarily true that the dual objective {\em decreases}; indeed, if $e_1$ and $e_2$ are almost ``equivalent and independent'' (that is, $f(\{e_1,e_2\}) \approx f(e_1) + f(e_2)$ and $f(e_1) \approx f(e_2)$), then $\dual(\tau^{(1)})$ can be significantly larger. The idea behind our primal-dual algorithm is that instead of {\em discretely} jumping from $\tau^{(0)}$ which puts all its mass on $\emptyset$ to $\tau^{(1)}$ which puts all its mass on $\{e_1\}$, we do a smoother transition always trying to decrease the dual value. With this intuition set, we now proceed to give details of the algorithm.

As mentioned above, we grow our set $\ALG$ in rounds and this naturally leads to a {\em chain} $\calA$ of subsets beginning with the empty set, the first element that is picked, the first two that are picked, and so on. The support of our dual solution will always be the sets in this chain, and the current primal solution, $\ALG$, will be the maximal set in this chain. As illustrated in \eqref{eq:dual-reinter}, 
we use $\gamma :=  \sum_{S\in \calA} \tau_S f(S)$, $\beta_j := \sum_{S\in \calA} \tau_S f_S(j)$, and $\alpha := \max_{j\in V} \beta_j$, and $\dual = k\alpha + \gamma$. All these variables depend on $\tau$ and will change when $\tau$ changes. Finally, we maintain a subset of ``tight'' elements: $T := \{j\in V~:~\alpha = \beta_j\}$; these are the set of elements with the largest expected marginal value. We initially set  $\calA = \{\emptyset\}$ and $\tau_\emptyset = 1$. Note that all other variables are defined by these. 
In particular, $\gamma = 0$, $\beta_j = f(\{j\})$ for all $j\in V$, $\alpha = \max_{j\in V} f(\{j\})$, $T = \{j\in V:f(j) = \alpha\}$, and 
$\dual = k\alpha$.

We begin a round attempting to modify the distribution $\tau$ so as to decrease $\dual(\tau)$. We do so by trying to move mass from the smaller sets in the chain $\calA$ to the 
maximal set in the chain which, recall, is our current solution $\ALG$. If we fail to decrease the dual objective value, then we pick a particular element with the largest expected marginal with respect to $\tau$, that is $\beta_j = \Exp_{S\sim \tau}[f_S(j)]$, into the primal set, and add the set $\ALG\cup \{j\}$
to the chain $\calA$ and move back to the dual-progress phase. Before we precisely describe how the $\tau$'s are modified, let us attempt to qualitatively explain why 
this method is a good idea. 

 Consider what happens to the dual objective when we move from $\tau$ to $\tau'$ where some mass is moved to the maximal set $\ALG\in \calA$. 
The term $\gamma' = \Exp_{S\sim \tau'} [f(S)] \geq \gamma = \Exp_{S\sim \tau} [f(S)]$ just by {\em monotonicity} of the function $f$. So, the second term in the objective~\eqref{eq:dual-reinter} actually increases when we go from $\tau$ to $\tau'$.
On the other hand, {\em submodularity} 
dictates that for any element $j$, $\beta'_j = \Exp_{S\sim \tau'}[f_S(j)] \leq \beta_j$, and so $\alpha' \leq \alpha$. So, the first term in the objective~\eqref{eq:dual-reinter}, which is multiplied by $k$, decreases.
If this drop, $k\cdot(\alpha - \alpha')$, cannot cancel the increase $\gamma' - \gamma$, then it must be that there is a tight element $j$ whose marginal to $\ALG$ is 
``similar'' to the marginals of the smaller sets in $\calA$; put differently, $j$ doesn't exhibit the decreasing marginal utility (at least not to a large extent).
This prompts us picking $j$ into our solution, explaining our primal pick. We now give the precise details and hope this paragraph helps the reader. \smallskip

We describe the dual modification as a continuous-time process, and in particular, assert the rate of change of the $\tau_S$'s over time. This leads, by definition, to a change in the values of $\gamma$, $\beta_j$'s, and $\alpha$. Using a notation from physics, we denote these rates by adding a dot above the variable name.
The dual modification phase is run {\em only if} the dual decreases. 
\begin{itemize}
	\item We set $\dt{\tau_{\ALG}} = 1 - \tau_{\ALG}$ and $\dt{\tau_S} = - \tau_S$ for all other $S\in \calA \setminus \ALG$. That is, we try and increase the mass on the maximal set 
	$\ALG$ and decrease the mass on all the other ``inner'' sets. We do so in a ``multiplicative'' fashion, in that, the maginitude of the rate of change depends on the current mass of the set.
	As noted above, the change in $\tau$-variables leads to changes in all other variables which we discuss below.
	
	\item  Note that the $\gamma = \sum_{S \in \calA} \tau_S f(S)$ variable changes as follows.
	\begin{equation}\label{eq:gamma-change}
		\dt{\gamma} = (1 - \tau_{\ALG})f(\ALG) - \sum_{S\in \calA\setminus \ALG} \tau_S f(S) = f(\ALG) - \gamma
	\end{equation}			
Since $f$ is monotone, this quantity will be non-negative, and thus, as we explained earlier, we increase the $\gamma$ variable over time. 
	\item 
	Note that the $\beta_j = \sum_{S \in \calA} \tau_S f_S(j)$ variables change as follows.
	\begin{equation}\label{eq:beta-change}
		\forall j\in V,~~~ \dt{\beta_j} = (1 - \tau_{\ALG})f_\ALG(j) - \sum_{S\in \calA\setminus \ALG} \tau_S f_S(j) =  f_\ALG(j) - \beta_j
	\end{equation}
	Since $f$ is submodular $f_S(j) \geq f_\ALG(j)$ for all $S\in \calA$ since $\calA$ is a chain, and thus, as we explained earlier, $\dt{\beta_j}$ will be non-positive. That is, these variables decrease in value.
	
	\item The $\alpha$ variables changes as follows.
	\begin{equation}\label{eq:alpha-change}
		\dt{\alpha} = \max_{j\in T} \dt{\beta_j}
	\end{equation}
	The rate of change of $\alpha$ is dictated by the highest rate of $\beta_j$ change among the tight elements. Again, when $f$ is submodular these $\dt{\beta_j}$'s are negative, and thus we are going to decrease the value of $\alpha$. 
	However, note that as we keep changing $\alpha$'s and $\beta_j$'s, the set $T$ may change, and in particular a new element may appear in $T$ and lead to a different $\dt{\alpha}$. We need to keep track of this. 
	
	\item Once $\dt{\gamma}$ and $\dt{\alpha}$ are defined, we see that $\dt{\dual} = k\dt{\alpha} + \dt{\gamma}$. We only do the changes to the previous steps as long as $\dt{\dual} < 0$. Otherwise, we stop this dual-modification phase. At the beginning, note that $\dt{\gamma} = 0$, $\dt{\beta_j} = 0$, and $\dt{\alpha} = 0$ implying $\dt{\dual} = 0$. Therefore, the dual-modification phase does not run at the beginning, and we move to the primal-choice phase. 
\end{itemize}

\noindent
When we cannot decrease the dual objective, we make a primal pick. As mentioned above, we consider the elements with the largest expected marginal, and in particular
this is the set of tight elements $T = \{j\in V~:~\beta_j = \alpha\}$. There can be multiple elements in this set $T$, and we pick the element $j$ with the {\em largest}
$\dt{\beta_j} = f_{\ALG}(j) - \beta_j$; that is, the element $j$ whose rate of change was the same as the rate of change of $\alpha$.
Note that for any two $j$ and $j'$ in $T$, we have $\beta_j = \beta_{j'} = \alpha$, so an alternate way would be to say we pick ``greedily'' from $T$: we pick 
the $j\in T$ with the largest $f_{\ALG}(j)$. Once we identify $j$, we add it to $\ALG$, and this new $\ALG$ into the chain $\calA$, and move back to the dual-modification phase.
In the new dual modification phase, the algorithm would try to raise the $\tau$ value of the new set $\ALG$ which currently has no $\tau$-mass on it. 
The above two dual-primal processes continue till we pick $k$ elements in $\ALG$ at which point we terminate. 
\ipco{Due to space restrictions, we describe the full pseudocode and an illustrative example in~\Cref{sec:illus-ipco}.}
\arxiv{The full pseudocode of the algorithm is described below. 
		\begin{center}
			\footnotesize{
			\fbox{%
				\begin{varwidth}{\dimexpr\linewidth-2\fboxsep-2\fboxrule\relax}
						\begin{algorithmic}[1]
							\Procedure{MSM Primal-Dual}{$f,V,k$}: 
							\State $i\eq 1$; $\calA^{(i)} \gets \{\emptyset\}$; $\ALG^{(i)} \leftarrow \emptyset$; $\tau_\emptyset \gets 1$; $\gamma \gets 0$
							\State $\forall j \in V, \beta_j \gets f(\{j\}), \dt{\beta}(j) \gets 0$; $\alpha \gets \max_{j \in V}f(j)$; $\dt{\alpha} \gets 0$; $\dt{\gamma} \gets 0$
							\State $T \gets \{{j \in V \mid \beta_j = \alpha}\}$ 
							\While{$i < k+1$}
							\While{$k\dt{\alpha} + \dt{\gamma} < 0$} 
							\LineComment{\sc Make Dual Progress}
							\State $\dt{\tau_{\ALG^{(i)}}} \eq 1 - \tau_{\ALG^{(i)}}$. \label{alg:mdp-first}
							\State $\dt{\tau_{S}} \eq - \tau_{S}$, $\forall S\in \calA^{(i)}\setminus \ALG^{(i)}$.
							\State $\dt{\beta}_j \eq f_{\ALG^{(i)}}(j) - \beta_j$, $\forall j\in V$. 
							\State $\dt{\gamma} \gets f(\ALG^{(i)}) - \gamma$
							\State $T \gets \{{j \in V \mid \beta_j = \alpha}\}$ \Comment{$j$'s may leave and enter $T$ as we change $\alpha$, $\beta_j$'s}
							\State $\dt{\alpha} \gets \max_{j\in T}(\dt{\beta}_j)$ \label{alg:mdp-last}
							\EndWhile
							\LineComment{\sc Make Primal Pick}
							\State $p_i \eq \arg \max_{j\in T} \dt{\beta_j}$; $\ALG^{(i+1)} \eq \ALG^{(i)} \cup p_i$;  $\calA^{(i+1)} \eq \calA^{(i)} \cup \ALG^{(i+1)}$ \label{alg:16}
							\State $\dt{\gamma} \eq f(\ALG^{(i+1)}) - \gamma$. \label{alg:17}
							\State $\dt{\beta_j} \eq f_{\ALG^{(i+1)}}(j) - \beta_j$, $\forall j\in V$. 
							\State $\dt{\alpha} \eq \max_{j\in T} \dt{\beta_j}$; $i\eq i+1$. \label{alg:19}
							\EndWhile
							\State \Return $\ALG^{(k+1)}$, $(\alpha, \tau_S~:~S\in \calA^{(k+1)})$.
							\EndProcedure
						\end{algorithmic}					
				\end{varwidth}%
			}
		}
		\end{center}
		As written, lines~\ref{alg:mdp-first} to~\ref{alg:mdp-last} continuously checks the while-loop condition. In~\Cref{sec:discrete} we describe how to implement this efficiently.
		In short, the only times we need to be careful is when a new element enters the set $T$ and these ``times'' can be found via solving simple ordinary differential equations. }
 \arxiv{	\subsubsection*{Illustrative Example} \label{subsec:illexample}
We provide an example which illustrates the working of the above algorithm and also use this to show how the above differs from the greedy algorithm.
The submodular function in this example is a coverage function. The universe is $\{A,B,C,D\}$ which correspond to sets described in the picture below, and the parameter $k = 3$.
Given a subset $S\subseteq \{A,B,C,D\}$, the function $f(S)$ is the number of points covered by the sets in $S$. So, for instance, $f(\{A\}) = 3$ and $f(\{D\}) = 2$, 
$f(\{A,B,D\}) = 7$, and $f(\{A,C,D\}) = 6$. 

\begin{figure}[H]
	\centering
	\includegraphics[width=0.6\textwidth]{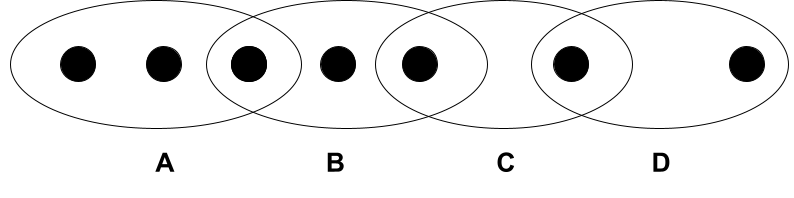}
	\label{fig:pd-better}
\end{figure}
Let us first consider what the greedy algorithm picks. Assuming we break ties adversarially\footnote{That is, we pick $A$ then $C$ then $D$; this ordering for greedy can be accomplished with a slight modification on the weight of $B$.}, the greedy algorithm's sets are $\{\emptyset, \{A\}, \{A,C\}, \{A,C,D\}\}$, which gives a final value of $6$. The optimum solution in this instance is $\{A,B,D\}$ giving an $f$-value of $7$.

\arxiv{Before we go to our algorithm, let us consider the dual solutions which put all their masses on the $4$ sets above. For instance if all the mass is put on the empty set, the dual value is $k$ times the maximum marginal which is $3\times 3 = 9$. Similarly, if $\tau$ puts the full mass on $\{A\}$, we get $\dual(\tau) = 3\times 2 + 3 = 9$ where the ``$2$'' arises because this is $C$'s marginal over $\{A\}$. Thus, the dual doesn't decrease (as we note in the second paragraph of the previous section, the dual could have even increased). The dual does decrease if 
	$\tau$ puts all the mass on $\{A,C\}$: we get $\dual(\tau) = 3\times 1 + 5 = 8$. And so, the ``dual-fitting'' of the greedy solution cannot give a dual solution less than $8$.}

Now let us understand what the primal-dual algorithm does. Observe that the first pick of the primal-dual algorithm is the same as the greedy algorithm, in this case, the element $A$ is picked in the solution $\ALG$. At this time, the chain is $\{\emptyset, \{A\}\}$. The dual now tries to move some mass from the empty set onto the set $\{A\}$.
We first observe that this step fails in increasing the dual. To see this, note that if we move to $\tau'$ which puts $(1-\eps)$ mass on $\emptyset$ and $\eps$ on $\{A\}$, 
for the element $B$, we get $\Exp_{S\sim \tau'}[f_S(B)] = 3\cdot(1-\eps) + 2\cdot \eps = 3 - \eps$. Furthermore, $\Exp_{S\sim \tau'}[f(S)] = 0\cdot (1-\eps) + 3\cdot \eps = 3\eps$.
And thus, the dual value is, in fact, $3\times  (3-\eps) + 3\eps = 9$ implying the dual doesn't decrease with such a move. So, the primal algorithm picks the element $j$ with the largest $\Exp_{S\sim \tau} [f_S(j)]$.
Note, however, that $\tau$ is still the dual with all its mass on $\emptyset$, and thus the primal-dual algorithm picks the element $B$ {\em instead} of element $C$ (which greedy picks).
This is where, for this example, the primal-dual algorithm and the greedy algorithm differ. \smallskip

Let's see what occurs next. Now, the chain is $\{\emptyset, \{A\}, \{A,B\}\}$. Now let's see what happens when we move $\eps$-mass from $\emptyset$ to $\{A,B\}$.
First, we see that $\Exp_{S\sim \tau'}[f(S)] = 0\cdot (1-\eps) + 5\eps = 5\eps$. Now let's see how $\Exp_{S\sim \tau'} [f_S(j)]$ behaves. We see that
\[
\Exp_{S\sim \tau'}[f_S(j)] = \begin{cases}
	3\cdot(1-\eps) + 0\cdot \eps = 3 - 3\eps& \text{for}~j = A \\
	3\cdot(1-\eps) + 0\cdot \eps = 3 - 3\eps & \text{for}~j = B \\
	2\cdot(1-\eps) + 1\cdot \eps = 2 - \eps & \text{for}~j = C \\
	2\cdot(1-\eps) + 2\cdot \eps = 2 & \text{for}~j = D \\
\end{cases}
\]		
And so, the dual value is $3 \times \max(3-3\eps, 2) + 5\eps$, and indeed this {\em does} drop below $9$ for small $\eps$: it drops till $\eps = 1/3$ and the dual reaches the value $7\frac{2}{3}$, and at that point the element $D$ enters the tight set. Note, we have obtained a dual solution whose objective is better than the ``greedy dual''. Indeed, our experiments in~\Cref{sec:experiment} will show that the dual our algorithm generates is often much smaller than what the greedy dual is.
Coming back to this instance, our algorithm makes the primal picks $D$ and gives, for this instance, the optimum solution which is $\{A,B,D\}$. The algorithm terminates.

We finish this discussion by noting that, although we do not add it, one could apply another round of ``make-dual progress'' after picking $k$ elements. For this particular instance,
we could think what occurs when we move mass from $\emptyset$ and $\{A,B\}$ to the set $\{A,B,D\}$; turns out, we can move all the mass on the set $\{A,B,D\}$ and this gives a dual value of $7$
which equals the primal.}

%
%
%
%
		
		\subsection{Analysis}
		
		The final $(\alpha, \tau_S)$'s form a feasible dual with objective value equal to $\dual = k\alpha + \gamma$. We now show that $f(\ALG^{(k+1)}) \geq \eta_k \cdot (k\alpha + \gamma)$ where $\eta_k = 1 - \left(1 - \frac{1}{k}\right)^k$. 
		\begin{mdframed}[backgroundcolor=gray!20,topline=false,bottomline=false,leftline=false,rightline=false] 
			\begin{theorem}
				The {\sc MSM Primal-Dual} algorithm returns a set $\ALG^{(k+1)}$ and a feasible dual solution $\tau$ with objective $\dual$ with the property
				\[
				f(\ALG^{(k+1)}) \geq \left(1 - \left(1 - \frac{1}{k}\right)^k \right) \cdot \dual
				\]
			\end{theorem}
		\end{mdframed}
		\begin{proof}
			To this end, fix time just {\em before} the $i$th item $p_i$ is picked, 
			In what follows, all the other variables and their rates are precisely at this time.
			Let $\Delta^{(i)}$ refer to the increase in $f(\cdot)$ from picking $p_i$, that is, 
			\[
			\Delta^{(i)} := f_{\ALG^{(i)}} (p_i)
			\] 
			Now observe that 
			\begin{equation}\label{eq:main}
				f_{\ALG^{(i)}}(p_i) \geq \beta_{p_i} - \frac{\dt{\gamma}}{k}
			\end{equation}
			The reason is $p_i = \arg \max_{j\in T} \dt{\beta_j}$ and $\dt{\alpha} = \max_{j\in T} \dt{\beta_j}$. Therefore, $k\dt{\beta_{p_i}} + \dt{\gamma} = k\dt{\alpha} + \dt{\gamma} \geq 0$, since
			we couldn't decrease dual anymore. Now using \eqref{eq:beta-change}, we know that $\dt{\beta_{p_i}} = f_{\ALG^{(i)}}(p_i) - \beta_{p_i}$. Rearranging gives us \eqref{eq:main}. 
			
			Next, we use the fact that $p_i \in T$ to deduce $\beta_{p_i} = \alpha$, and use \eqref{eq:gamma-change} to note $\dt{\gamma} = f(\ALG^{(i)}) - \gamma$. Substituing above, we get
			\[
			f(\ALG^{(i+1)}) - f(\ALG^{(i)}) = f_{\ALG^{(i)}}(p_i)  \geq \alpha - \frac{f(\ALG^{(i)} - \gamma)}{k} 
			\]
			rearranging which implies
			\begin{equation}
				f(\ALG^{(i+1)}) \geq \left(1 - \frac{1}{k}\right) f(\ALG^{(i)}) + \frac{1}{k}\cdot \dual
			\end{equation}
			Since $\dual$ just before $i$th item is picked is only at least the final dual objective (since we only decrease dual).~\Cref{clm:ana} gives us the theorem.
		\end{proof}
\arxiv{
\subsection{Discretization and Reformulation}\label{sec:discrete}
		
		In this step we describe the discrete reformulation of the continuous algorithm described above, especially the lines~\ref{alg:mdp-first} to~\ref{alg:mdp-last}. To this end, 
		suppose we have just added $p_i$ to $\ALG^{(i)}$ in Line~\ref{alg:16}, and the resulting $\dt{\gamma}$ and $\dt{\alpha}$ calculated in Line~\ref{alg:17} and Line~\ref{alg:19}, respectively,
		satisfy $k\dt{\alpha} + \dt{\gamma} < 0$. This is when we enter the inner while-loop and run lines~\ref{alg:mdp-first} to~\ref{alg:mdp-last}.
		Let this be $t=0$ and let all variables and their rates be in indexed by this time, that is, $\alpha(0) = \alpha$, $\dt{\alpha}(0) = \dt{\alpha}$, and so on and so forth.
		We also let $T(0)$ denote the tight set at this point. 

		We now explain evolution of these various variables as a function of $t$ while we are in the inner while-loop. Since $\calA$ and $\ALG$ don't change in the inner while loop all the rates behave in a predictable manner which will allow us to ``shortcut'' the inner while loop. We discuss this, and then describe the discrete algorithm.
		
		For brevity's sake, we use $F:= f(\ALG)$ and $F_j := f_{\ALG}(j)$ for all $j\in V$. The variable $\dt{\beta}_j(t)$ behaves as per the ordinary differential equation
		\begin{equation}\label{eq:beta-evol}
                \ipco{\small}
			\forall j \in V,~~~ \dt{\beta_j}(t) = F_j - \beta_j(t) ~~\Rightarrow~~~~~ \beta_j(t) = \beta_j(0)\cdot e^{-t} + F_j\cdot (1-e^{-t}) ~~\textrm{and}~~ \dt{\beta}_j(t) = \dt{\beta}_j(0)\cdot e^{-t}
		\end{equation}
		Similarly, we get the evolution of $\gamma(t)$ and $\dt{\gamma}(t)$ as 
		\[
		\dt{\gamma}(t) = F - \gamma(t)~~\Rightarrow~~~~~ \gamma(t) = \gamma(0)\cdot e^{-t} + F\cdot (1-e^{-t})~~~\textrm{and}~~~ \dt{\gamma}(t) = \dt{\gamma}(0)\cdot e^{-t}
		\]
		That is, the $\beta_j$ (and $\gamma$) variables are a convex combinations of $F_j$ (of $F$) and their old values. Crucially note that the ratio $\dt{\beta}_j(t)/\dt{\gamma}(t)$ is unchanged over time
		and is the original $\dt{\beta_j}(0)/\dt{\gamma}(0)$.
		
		Next, note that the value of $\dt{\alpha}(t) = \max_{j\in T(t)} \dt{\beta}_j(t)$ and thus equals $\dt{\beta}_{j_t}(t)$ for some ${j_t}\in T(t)$. Therefore, while $T(t)$ remains the same,
		the rate $\dt{\alpha}(t) = \dt{\beta}_{j_t}(t)$, and thus using the previous paragraphs argument, $\dt{\alpha}(t)/\dt{\gamma}(t)$ remains the same. In particular, till $T(t)$ changes,
		the while condition holds.
		
		Now note that until a new element is {\em added} to $T(t)$, the limiting $j_t$ never leaves $T(t)$, since $\beta_{j_t}(t)$'s rate of change and $\alpha$'s rate of change is the same as $\alpha(t)$'s rate of change.  Therefore, $\dt{\alpha}(t)$ changes {\em only} when a new element $\ell\notin T(t)$ is added to $T(t)$ whose $\dt{\beta}_\ell(t) > \dt{\beta}_{j_t}(t)$. At that instant, the value of $\dt{\alpha}(t)$ is set to $\dt{\beta}_\ell(t)$. At this point, we may have to check the while-loop condition.

		But we can {\em pre-compute} the set of element $\ell$ which would lead to a violation of the while-loop condition. This is again because the ratio $\dt{\beta}_\ell(t)/\dt{\gamma}(t)$ remains the same throughout. In particular, define the set
		\[
		W := \{j \in ~V~: k\dt{\beta}_j(0) + \dt{\gamma}(0) \geq 0\} ~\underbrace{=}_{\text{constant $\dt{\beta}_j(t)/\dt{\gamma(t)}$}}~~~ \{j \in ~V~: k\dt{\beta}_j(t) + \dt{\gamma}(t) \geq 0\}
		\]
		where $\beta_j$ and $\gamma$ are values at time $0$. For a sanity check, note that $W$ doesn't contain a tight element since $\dt{\alpha} = \max_{j\in T} \dt{\beta}_j$ and if there was an intersection we would violate the inner while loop condition. The discussion preceding the set definition leads us to conclude that the inner while loop terminates the first time $t$ when an element in the set $W$ enters the set $T(t)$. Next, let's figure out the time when an element $\ell$ not in the tight set enters it. \smallskip
		
		Fix an element $\ell \notin T(0)$ and let $\dt{\alpha}(0) = \dt{\beta}_{j_0}(t)$ for $j_0 \in T(t)$. If these were the only two elements, then we can find the $t$ at which $\ell$ would enter 
		$T(t)$ by simply using \eqref{eq:beta-evol} on $\ell$ and $j_0$. Namely,
		\[
		\beta_\ell(t)e^{-t} + F_\ell(1- e^{-t}) = \beta_{j_0}(t) e^{-t} + F_{j_0}(1 - e^{-t}) ~~\Rightarrow ~~~ t_\ell  = \ln\left(1 + \frac{\beta_{j_0} - \beta_\ell}{F_\ell - F_{j_0}}\right)
		\]
		
		In particular, the first $\ell \notin T(0)$ which enters $T(t)$ is the element with minimum $t_\ell$, namely
		\begin{equation}
			\ell = \arg \max_{j\notin T(0)} \left(\frac{F_{j} - F_{j_0}}{\beta_{j_0} - \beta_j}\right) ~~~=:j_1 \label{eq:discrete-ell}
		\end{equation}
		So, $\ell$ is the element whose marginal increase on $\ALG$ is larger than $j_0$'s, but the algorithm shades this by the denominator which is the difference of the $\beta$-values.
		At time $t_\ell$, this element $\ell$ enters the tight set at time $t = t_\ell$ and, as discussed before, $\dt{\alpha}(t)$ increases to $\dt{\beta}_\ell(t)$. If $\ell \in W$, the inner while loop terminates. Otherwise, we again use equation \ref{eq:discrete-ell} to find the next element to become tight with the only change being that $\ell$ takes the value of ``$j_0$'' as the element which dictates $\dt{\alpha}_t$. This process continues till we find a $W$-element. We write the pseudocode below.
		
		\begin{center}
			\fbox{%
				\begin{varwidth}{\dimexpr\linewidth-2\fboxsep-2\fboxrule\relax}
						\begin{algorithmic}[1]
							\Procedure{Discrete Make Dual Progress}:
							\LineComment{Replaces Lines~\ref{alg:mdp-first} to~\ref{alg:mdp-last} in Continuous Algorithm} 
							\LineComment{Below, we use $F_j := f_{\ALG^{(t)}}(j)$ for all $j$ and $F := f(\ALG^{(t)})$ for brevity}
							\State $j^* \eq \arg \max_{j\in T} \dt{\beta_j}$.
							\State $\ell \eq \arg \max_{j\in V\setminus T} \left(\frac{F_j - F_{j^*}}{\beta_{j^*} - \beta_j}\right)$ \Comment{Will enter $T$ next} \label{alg:5}
							\State $t \eq \ln\left(1 + \frac{\beta_{j^*} - \beta_\ell}{F_\ell - F_{j^*}}\right)$
							\State $\forall j \in V, \beta_j \eq \beta_j e^{-t} + F_j(1-e^{-t})$; $\dt{\beta_j} \eq F_j - \beta_j = \dt{\beta_j}e^{-t}$.
							\State $\gamma \gets \gamma e^{-t} + F(1 - e^{-t})$; $\dt{\gamma} \eq F - \gamma = \dt{\gamma}e^{-t}$.
							\State $\alpha \eq \beta_\ell e^{-t} + F_\ell (1-e^{-t})$.
							\State $T \gets \{{j \in V \mid \beta_j = \alpha}\}$ \Comment{$\ell$ enters $T$}
							\State If $\ell \in W$, then {\bf break}, else $j^* \eq \ell$ and go to line~\ref{alg:5}
							\EndProcedure
						\end{algorithmic}					
				\end{varwidth}%
			}
		\end{center}

		\paragraph{Query and Runtime Analysis.} 
		The number of queries is clearly $O(nk)$ since to make the algorithm run we need $f_{\ALG^{(t)}}(j)$ for all $j\in V$, and there are $k$ different sets.
		The running time is dominated by {\sc Discrete Make Dual Progress}. Naively, to find the element $\ell$ in Line~\ref{alg:5}, one can take $O(n)$ time, and 
		since any element can enter $T$ only once, this gives an $O(n^2)$ time for this procedure, and a total time of $O(kn^2)$ overall. 
		It is possible that techniques for speeding up the greedy algorithm as in~\cite{Minoux77,BadaV14,WeiIB14,MirzBKVK15,BreuBS20} may be applicable to speed the above algorithm as well, 
		and we leave this as an open direction. 
  
}
\subsection{A Dual-Fitting Analysis of the Greedy Algorithm}\label{sec:dual-fitting}
In this subsection, we show a simple ``dual-fitting'' proof that the greedy algorithm achieves a $\eta_k$-approximation. 
Recall, the greedy algorithm solves
$\max_{S:|S|=k} f(S)$ by picking elements in $k$ iterations. 
At the beginning of the $t$th iteration, it has a set $S_t$ of size $t-1$, with $S_1$ initialized to $\emptyset$.
At the $t$th iteration, the algorithm picks element $j_t$ with largest $f_{S_t}(j)$ and sets $S_{t+1} = S_t ~\cup~ j_t$.
The returned solution is $S_{k+1}$. 
We now describe a ``dual-fitting'' analysis of the greedy algorithm. We first describe $k$ different feasible dual solutions and show that $f(S_{k+1})$ is precisely $\eta_k$ times a certain weighted average, 
where $\eta_k = 1 - (1-1/k)^k$. This would prove 
\begin{theorem}
	$f(S_{k+1}) \geq \eta_k \cdot \dual$.
\end{theorem}
\begin{proof}
	The $k$ dual solutions correspond to the $k$-sets $S_1$ to $S_k$. 
	The $t$th dual feasible solution assigns a point-mass on the set $S_t$. So,
	\[
	\tau^{(t)}_{S_t} = 1~~\textrm{and all other $\tau^{(t)}_S = 0$}; ~~~ \alpha^{(t)} = \max_j f_{S_t}(j);~~~~~\dual^{(t)} = k\alpha^{(t)} + f(S_t)
	\]
	Crucially note by the design of the greedy algorithm:
	\[
	f(S_t) = \sum_{s=1}^{t-1} \alpha^{(s)}~~~\textrm{and, in particular,}~~~ f(S_{k+1}) = \sum_{t=1}^k \alpha^{(t)}
	\]
	and thus,
	\[
	\dual^{(t)} = k\alpha^{(t)} + \sum_{s=1}^{t-1} \alpha^{(s)}
	\]
	Now consider the fractions $\theta_1, \ldots, \theta_k$ defined as
	\[
	\theta_{t} = \left(1-\frac{1}{k}\right)\theta_{t+1}, ~~\textrm{and}~~\sum_{t=1}^k \theta_t = 1
	\]
	Solving for this implies $\theta_k = \frac{1}{k\eta_k}$ where $\eta_k = 1 - (1-1/k)^k$. \smallskip
	
	\noindent
	The geometric decay among the $\theta_t$'s implies that for any $1\leq t\leq k$,
	\begin{equation}\label{eq:boo}
		k\theta_t + \sum_{s > t} \theta_s = k\theta_{t+1} + \sum_{s > t+1} \theta_s = k\theta_k = \frac{1}{\eta_k}
	\end{equation}
	Next observe
        \ipco{\small}
	\[
	\sum_{t=1}^k \theta_t\dual^{(t)} = 		\sum_{t=1}^k \theta_t \left(k\alpha^{(t)} + \sum_{s=1}^{t-1} \alpha^{(s)}\right) = \sum_{t=1}^k \alpha^{(t)}\cdot\left(	k\theta_t + \sum_{s > t} \theta_s \right) \underbrace{=}_{\eqref{eq:boo}} \frac{1}{\eta_k} \cdot \sum_{t=1}^k \alpha^{(t)}
	\]
        \ipco{\normalsize}
	Noting that $f(S_{k+1}) = \sum_{t=1}^k \alpha^{(t)}$, we get that $f(S_{k+1}) = \eta_k \cdot \sum_{t=1}^k \theta_t \dual^{(t)}$, proving the theorem.
\end{proof}
We end by noting that one could also define the distribution $\tau_{\mathsf{Greed}}$ on sets which gave mass $\theta_i$ on the set $S_i$ for $1\leq i\leq k$.
By Jensen's inequality, we get that $\dual(\tau_{\mathsf{Greed}}) \leq \sum_{t=1}^k \theta_t\dual^{(t)}$, and therefore we would get $f(S_{k+1}) \geq \eta_k\cdot \dual(\tau_{\mathsf{Greed}})$
as well. In our experiments described in~\Cref{sec:experiment}, we use $\greedydual_1$ to denote the smaller of these two values. Furthermore, by considering only nonzero $\tau$ values on the sets $S_1,\cdots,S_{k+1}$, we can form a linear program using the dual as the sum described in $\eqref{eq:d1}$ which is no longer exponential in size. Note that the dual bound found from solving this minimization problem will be guaranteed to be at least as good as that of $\greedydual_1$, and in particular, will be within $(1-1/e)$. 
We denote this linear programming dual solution as $\greedydual_2$ in \Cref{sec:experiment}.
\section{Experiments} \label{sec:experiment}
We perform an empirical evaluation of our primal-dual algorithm. We run on several \msm~instances and compare (i) how our dual upper bound compares with the current state-of-the-art, namely the Balkanski-Qian-Singer (BQS) upper bound~\cite{BalkQS21}, (ii) how our primal solution compares with the greedy solution, and (iii) how our running times compare with the BQS upper bound algorithm. 

\paragraph{Submodular Functions.} We run on two types of submodular functions. One class is the example of coverage functions, where the set-system (bipartite graph) is obtained via random graphs created using four common graph generation models: Stochastic Block (SB) Model, Erdos-Reyni (ER), Watts-Strogatz (WS), and Barbasi-Albert (BA) with n$\approx$500 and k values from 10 to 40. The second class arises
from real-world optimization problems that are representative of different types of submodular maximization applications: optimizing sensor placement on California roadways \cite{CARoads}, revenue maximization for YouTube network data \cite{MirzBK}, influence maximization on Facebook network data \cite{TrauMP12}, influence maximization in citation networks \cite{Lesk07}, and optimizing recommendations on MovieLens 1M data \cite{Harper15}. Most of these data sets were obtained from the public website~\cite{Breuer} which was created along with the paper by Breuer, Balkanski and Singer~\cite{BreuBS20}. 

%

\paragraph{Upper Bound Benchmarks.} We compare the dual upper bound that our primal-dual algorithm returns (we call this $\primdual$) with two types of upper bound benchmarks. 
The first, $\bqs$, is the upper bound described by Balkanski, Qian, and Singer~\cite{BalkQS21}. In~\Cref{sec:bqs-dualfit} we give a short but self-contained description of their upper bound.
The main empirical finding of~\cite{BalkQS21} is that the $\bqs$ upper bound significantly out-performs other upper bound benchmarks, such as the {\em top-$k$} which is the sum of the $k$ largest singletons or the upper bound described in~\cite{SakaI18}. As we report below, $\primdual$ is quite comparable to $\bqs$ and out-performs these upper-bounds; for simplicity of presentation we do not
display these comparisons.

The second type of upper bound benchmark we compare with is the dual fitting certificates that we describe in~\Cref{sec:dual-fitting}, $\greedydual_1$ and $\greedydual_2$. Unlike $\bqs$, there is a worst-case guarantee on both $\greedydual$ certificates; by construction, we have that the greedy solution is at least $(1-1/e)\greedydual_1\geq (1-1/e)\greedydual_2$. We were interested in what this ratio is empirically. Note that the only guarantee~\cite{BalkQS21} 
could show on $\bqs$ is that the greedy solution is at least $\bqs/2$. 

\paragraph{Findings.} We report some representative trials and the interested reader may obtain the full results and implementation if each experiment from~\cite{CoteGithub}.
In the figure below, we show how the three different upper bounds compare on three different submodular functions. 
The $x$-axis is the various values of $k$ and the $y$-axis is the numeric value of the upper bound.
This picture shows how $\primdual$, $\bqs$, and $\greedydual_2$ are comparable while being significantly better than $\greedydual_1$ (recall, lower is better). 

	\includegraphics[width=0.9\textwidth]{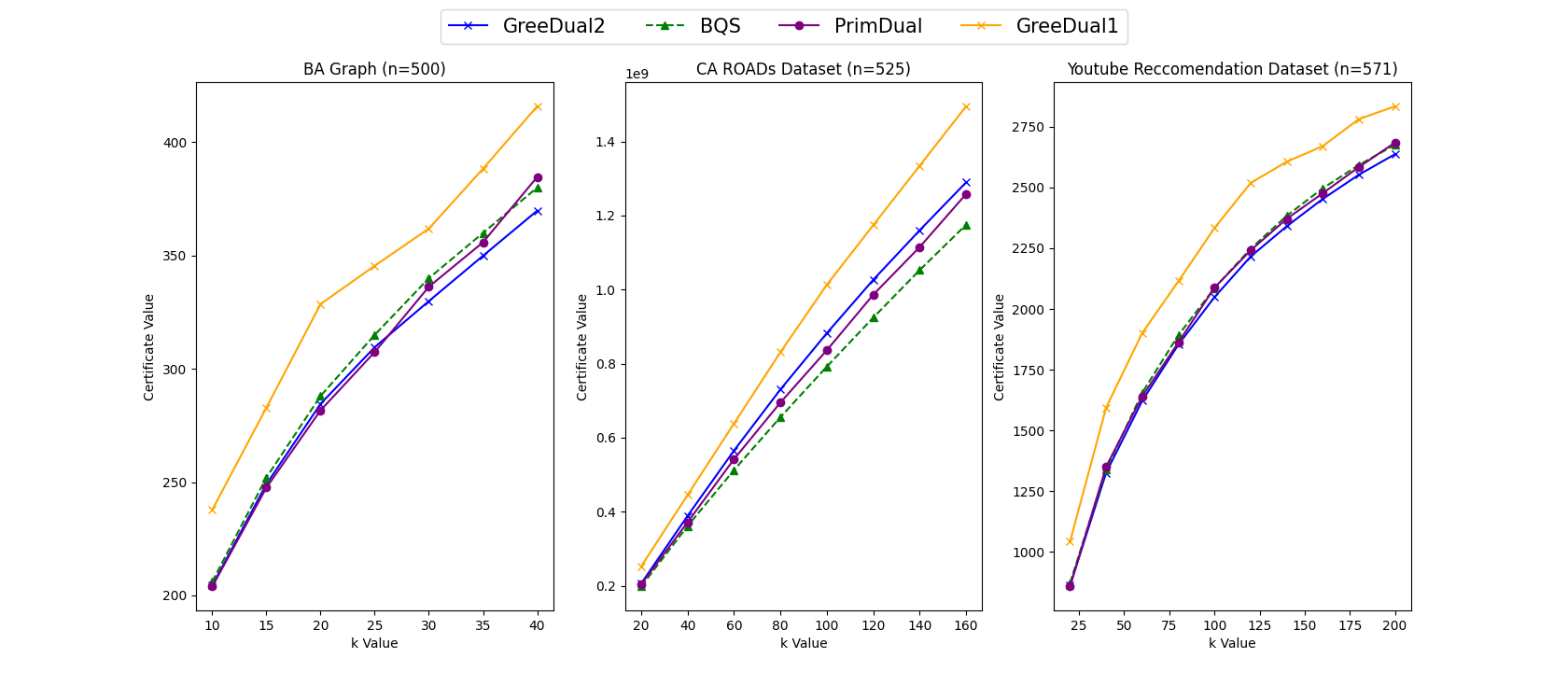}
\ipco{Due to space restrictions, we move the remaining empirical findings to~\Cref{sec:expt-ipco}}

\arxiv{\noindent
In the following table, we also give numerical values of the ratio of the greedy solution and the dual solution; closer to $1$ is better. The various columns 
are our submodular functions and the number is the average taken over all $k$'s. As one can see $\primdual$, $\bqs$, and $\greedydual_2$ often provide an approximation guarantee  well over 85\%.
\begin{table}[H]
    \small
	\centering
	\scalebox{0.9}{
		\begin{tabular}{|c c c c c c c c c c|} 
			\hline
			& BA Graph & ER Graph & SBM Graph & WS Graph &  CA Roads & Facebook & Youtube & Citations & MovieLens\\ [0.5ex] 
			\hline\hline
			$\primdual$ & 0.9589 & 0.8839 & 0.9530 & 0.9919 & 0.9167 & 0.8212 & 0.9381 & 0.7952& 0.9369\\ 
			\hline
			$\greedydual_1$ & 0.7530 & 0.7341 & 0.7709 & 0.8831 & 0.7645 & 0.7071 & 0.7181 & 0.6981 & 0.7852\\ 
			\hline
               $\greedydual_2$ & 0.9663 & 0.8272 & 0.9310 & 0.9919 & 0.8829 & 0.7935 & 0.9494 & 0.7685 & 0.9544\\ 
                \hline
			$\bqs$ & 0.9479 & 0.8792 & 0.9782 & 0.9944 & 0.9660 & 0.8309 & 0.9342 & 0.8276 & 0.9615\\
			\hline
		\end{tabular}
	}
	\end{table}

The following figure shows the comparison of the greedy solution with the primal solution returned by our primal-dual algorithm.
As can be see, the performance of the primal algorithm almost matches that of the greedy algorithm, although the latter
often \emph{slightly} outperforms the former with a maximum deviation of 8.1\% across all trials.
	
	\includegraphics[width=0.9\textwidth]{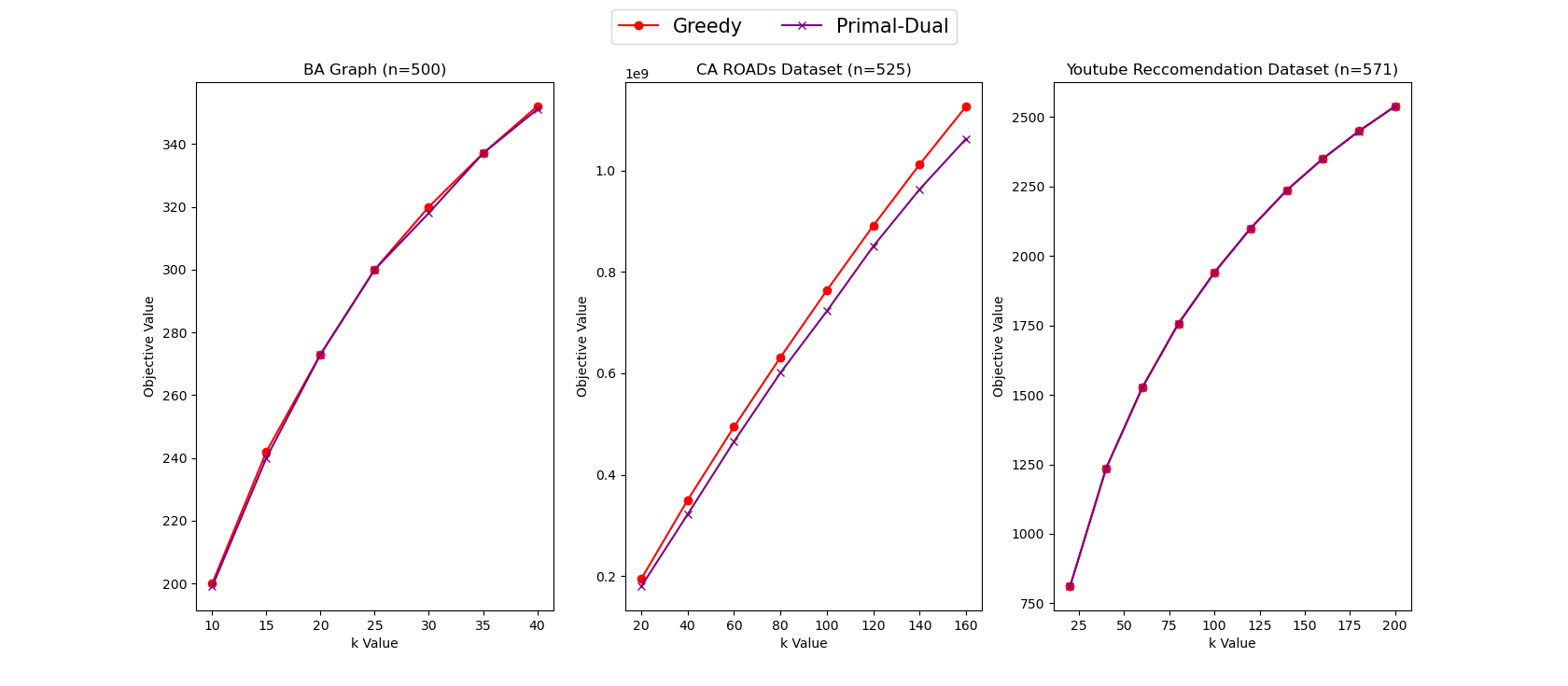}

The following figure shows the runtimes of our primal-dual algorithm, the (naive) implementation of the greedy algorithm, the ``Method 4'' of~\cite{BalkQS21} which returns $\bqs$, and a LP solution for $\greedydual_2$ which uses the SciPy interior-point implementation \cite{2020SciPy-NMeth}.
The primal-dual algorithm often runs slightly slower than the greedy algorithm, but often beats the BQS upper bound calculation algorithm. Note that this method needs to run the greedy algorithm
and then run the ``Method 3'' of~\cite{BalkQS21} (see~\Cref{sec:bqs-dualfit} for more on this algorithm). Note that the $\greedydual_2$ upper bound runs significantly slower than both the primal-dual algorithm and the BQS upper bound due to the computationally intensive nature of solving an LP.

	\includegraphics[width=0.9\textwidth]{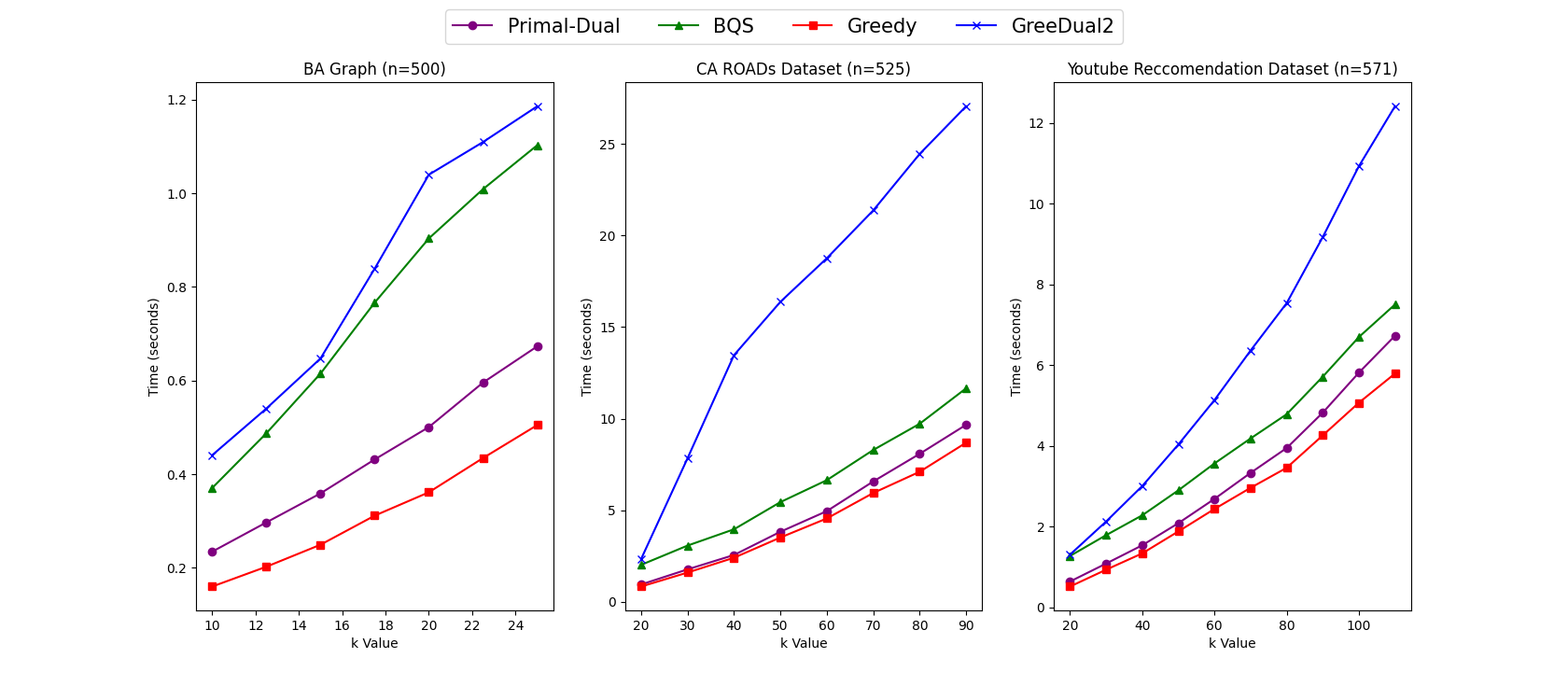}
}
\arxiv{\section{Dual Analyses for \mmsm.}\label{sec:matroid-dualfit}
In this section we expand our primal-dual exploration to the of \msm~where instead of a cardinality constraint on the set $S$ being returned, the constraint is that the set $S$
needs to be an independent set of a certain matroid $M(V, \mathcal{I})$. In short, we call this the \mmsm~problem.
This is a richer problem capturing many other fundamental problems such as the generalized assignment problem (GAP)~\cite{ShmoY93,ChekK05,FleiGMS11} from operations research and the maximum submodular welfare maximization problem~\cite{LehmLN01,DobzS06,FeigV06} from algorithmic economics. The \mmsm~problem was first studied by Fisher, Nemhauser, and Wolsey~\cite{FishNW78} soon after the cardinality case, and they proved that the natural greedy algorithm which proceeds in rounds picking an {\em eligible} element (which can be added without violating the independence) which leads to the maximum function increase, is a $1/2$-approximation. In a seminal result, Calinescu~\etal~\cite{CaliCPV11} described a $(1-1/e)$-approximation via the {\em continuous greedy} algorithm. Unlike the greedy algorithm, the continuous greedy algorithm is randomized, and it is an open question to design a deterministic algorithm which achieves the optimal 
$(1-1/e)$-factor. While we do not generate such an algorithm in this section, we hope that the ideas and approaches explored here will aid in future research surrounding this question. Indeed, the work in this section is purely expositional and provides little novel results to the reader - rather this section should be taken as a collection of insights for those interested in further exploration of the \mmsm~problem.


We begin with the obvious generalization of the Nemhauser-Wolsey~\cite{NemhW81} linear programming relaxation which, as mentioned before, was also studied by Calinescu~\etal~\cite{CaliCPV07} in their IPCO version. Belore, $r_M(T)$ is the {\em rank} function which returns the size of the largest independent subset of $T$. It is well known that \eqref{eq:p3} captures
the integer hull of the characteristic vector of all independent sets. We also write the dual.
\arxiv{
\begin{minipage}{0.44\textwidth}
	\fontsize{10pt}{12pt}
	\begin{alignat}{6}
		\lp:=&&\max~~~~\lambda~&&~&\tag{\sf{primal-matroid}} \label{eq:primal-matroid} \\
		&&\sum_{j \in T} x_j &\leq && r_M(T), && \forall T \subseteq V~&\tag{P3}\label{eq:p3}\\
		&&\lambda - \sum_{j \in V} f_{S}(j)x_j &~\leq~ && f(S), && ~\forall S \subseteq V \tag{P4}\label{eq:p4}\\
		&&x_j &~\geq~&& 0, && ~\forall j\in V \notag
	\end{alignat}
\end{minipage} ~~~\vline
\begin{minipage}{0.52\textwidth}
	\fontsize{10pt}{12pt}
	\begin{alignat}{6}
		\min~   &&~\sum_{T \subseteq V}r_M(T) \alpha_T + \sum_{S\subseteq V}\tau_S f(S)  &=:&& \dual \notag\\
		&&\sum_{T \subseteq V :	 j\in T}\alpha_T - \sum_{S\subseteq V}\tau_S f_{S}(j) &~\geq~ && 0,~ \forall j \in V \tag{D4} \label{eq:d4}\\
		&& \sum_{S\subseteq V} \tau_S &~=~&& 1 \tag{D5} \label{eq:d5}\\
		&& \tau_S &~\geq~&& 0, ~\forall S\subseteq V \tag{D6} \label{eq:d6}
	\end{alignat}
\end{minipage}
}
\ipco{
	\begin{alignat}{6}
	\lp:=&&\max~~~~\lambda~&&~&\tag{\sf{primal-matroid}} \label{eq:primal-matroid} \\
	&&\sum_{j \in T} x_j &\leq && r_M(T), && \forall T \subseteq V~&\tag{P3}\label{eq:p3}\\
	&&\lambda - \sum_{j \in V} f_{S}(j)x_j &~\leq~ && f(S), && ~\forall S \subseteq V \tag{P4}\label{eq:p4}\\
	&&x_j &~\geq~&& 0, && ~\forall j\in V \notag
\end{alignat}
	\begin{alignat}{6}
	\min~   &&~\sum_{T \subseteq V}r_M(T) \alpha_T + \sum_{S\subseteq V}\tau_S f(S)  &=:&& \dual \notag\\
	&&\sum_{T \subseteq V :	 j\in T}\alpha_T - \sum_{S\subseteq V}\tau_S f_{S}(j) &~\geq~ && 0,~ \forall j \in V \tag{D4} \label{eq:d4}\\
	&& \sum_{S\subseteq V} \tau_S &~=~&& 1 \tag{D5} \label{eq:d5}\\
	&& \tau_S &~\geq~&& 0, ~\forall S\subseteq V \tag{D6} \label{eq:d6}
\end{alignat}
}

\noindent
As in the case of \msm, the dual contains a {\em probability} distribution $\tau$ over subsets of $V$. The $\alpha_T$'s are the dual variables on subsets, and these correspond to 
the rank-constraints, \eqref{eq:p3}, of the primal. These are harder to give intuition about, but are precisely the dual-variables that arise, for instance, in 
Edmond's (see, for instance,~\cite{Edmonds03}) proof that the ``greedy'' algorithm finds the maximum weight basis in a matroid. 
The objective function is the expected function value with respect to $\tau$ plus the rank-weighted $\alpha$-mass over the subsets.
The constraint on $\alpha_T$'s are that the total $\alpha$-mass facing an element $j$ must be at least the expected marginal of $j$ on a set drawn from the distribution $\tau$.	

\subsection{A Dual Fitting Analysis of the Greedy Algorithm for \mmsm}
The commonly used greedy algorithm for \msm~can be generalized under a matroid constraint by repeatedly choosing the element with the maximum marginal utility whose addition to the current solution set lies within the matroid polytope. The result from running such a greedy algorithm on an instance of \mmsm~can be used to generate the following dual certificate.

We begin with some useful definitions, we define $k:=r_M(V)$ to be the rank of the full matroid, $S_i \subseteq V$ be the greedy solution after $i$ choices, $p_i \in V$ be the $i$th element chosen by the greedy solution, and $T_i := \spann(S_i)$, that is, the set of elements $j$ such that $S_i \cup j$ is dependent (that is, not in $\calI$). Note that $T_a \subseteq T_b$ for $1\leq a < b \leq k$
and $r_M(T_i) = i$ for all $1\leq i\leq k$.

Consider the dual solution, $\dual_G$, where $\tau_{S_k} = 1$, that is, we put all the mass on the greedy algorithm's output. Now, define the $\alpha$ values as follows: 
\[\alpha_{T_k} = \max_{j \in T_k}f_{S_k}(j)~~~~\textrm{and}~~~~
\alpha_{T_i} = \left(\max_{j \in T_i \setminus T_{i-1}}f_{S_k}(j)\right)-\sum_{\ell = i+1}^k\alpha_{T_\ell},~~\forall k > i \geq 1\]
\noindent
Now observe a couple of things
\begin{itemize}
	\item For every $e\in V$, there is an $1\leq i\leq k$ such that $e \in T_i, T_{i+1}, \ldots, T_k$ but $e\notin T_{i-1}$. And so,
	\begin{equation}\label{eq:goopy}
		\sum_{T \subseteq V : e\in T} \alpha_{T} = \sum_{\ell=i}^k \alpha_{T_\ell} \underbrace{=}_{\text{definition}} \max_{j\in T_i\setminus T_{i-1}} f_{S_k}(j) \geq f_{S_k}(e)
	\end{equation}
	where the inequality follows since $e\in T_i \setminus T_{i-1}$. Therefore, the $\alpha$'s and $\tau$'s satisfy \eqref{eq:d4}.
	\item Next, note by the definition of the $\alpha$'s, that 
	\begin{equation}\label{eq:bagha} 
		\sum_{i=1}^k i\cdot \alpha_{T_i} = \sum_{i=1}^k \left(\max_{j \in T_i \setminus T_{i-1}}f_{S_k}(j)\right) \underbrace{\leq}_{\text{submodularity}} \sum_{i=1}^k \left(\max_{j \in T_i \setminus T_{i-1}}f_{S_{i-1}}(j)\right)
	\end{equation}
	The LHS is simply $\sum_{T \subseteq V} r_M(T)\alpha_T$ and the RHS, by the greedy property of the algorithm, is precisely $f(S_k)$. This is because $\max_{j\in T_i\setminus T_{i-1}} f_{S_{i-1}}(j) = f(S_{i}) - f(S_{i-1})$, and the RHS summation telescopes to $f(S_k)$ (using $f(\emptyset) = 0$).
\end{itemize}
Together, we get that $\sum_{T \subseteq V} r_M(T)\alpha_T \leq f(S_k)$. By design, $\sum_{S \subseteq V} \tau_S f(S) = f(S_k)$. Therefore, the dual objective value of $(\alpha, \tau)$ is 
$\leq 2f(S_k)$, giving a dual-fitting proof of the fact that the greedy algorithm is a $2$-factor approximation. Note that if the inequalities in \eqref{eq:bagha} for a given is slack, we get a certifiably better than $2$-approximation.
%
%
%
It should also be noted that there exist instances which show the above analysis is tight.
\subsection{A Dual Fitting Analysis of the Continuous Greedy Algorithm}

The ``continuous greedy'' algorithm of  Calinescu \etal~in \cite{CaliCPV11} achieves an optimal  $(1-\frac{1}{e})$ approximation factor in expectation. We now show 
how one can also give a dual-fitting proof; more precisely, we describe a randomized algorithm which also returns a dual solution $(\alpha, \tau)$ such that 
$\Exp[f(R)] \geq (1-1/e)\dual(\alpha,\tau)$ where $R$ is the random output of continuous greedy. We remark that this is not a new result as such a result is implicit
in the ``certificate discussion'' present in Lemma 7.3 of \cite{ChekVZ18}. However, we believe that making this connection explicit may give some insights into possibly derandomizing
and potentially obtaining a {\em deterministic} $(1-1/e)$-approximation algorithm for the \mmsm~problem.
%

Before we describe the dual, let us give a (very) quick description of the continuous greedy algorithm from~\cite{CaliCPV11}; indeed, we assume the reader is familiar with this algorithm.
So, recall the multilinear extension, $F(y):=\sum_{S \subseteq U}\Pr_y[S] f(S)$ where $y\in [0,1]^n$ is a fractional point, and $\Pr_y[S]$ is the probability of drawing the set $S$ where element $j$ is picked in set $S$ with probability $y_j$. So, $\Pr_y[S] = \prod_{j\in S} y_j \prod_{j\notin S} (1-y_j)$. Calinescu~\etal~\cite{CaliCPV07} showed how given an $y\in \calP_M$, the matroid independent set polytope, one can get a distribution $\calD$ on independent sets of the matroid such that $\Exp_\calD[f(R)] \geq F(y)$ where $R$ is the output of the algorithm. Thus, it suffices to design an algorithm which obtains a $y^*\in \calP_M$ with $F(y^*) \geq (1-1/e)\opt$. This is the continuous greedy algorithm, and we describe the ``continuous version'' of it, and show how to construct a dual solution $(\alpha,\tau)$ such that $F(y^*) \geq (1-1/e)\dual$.

The continuous greedy algorithm starts with $y = \vec 0$ and continuously moves $y(t)$ in the direction of greatest increase within the feasible region until time $t=1$; the final solution $y^* = y(1)$.
The rate of increase of the $y(t)$-vector is given by 
\begin{equation}\label{eq:cont-greed}
	\frac{dy(t)}{dt} = \argmax_{v \in \mathcal{P}_M} ~~~v\boldsymbol{\cdot}\grad F(y(t)) \tag{Continuous Greedy}
\end{equation}	
where the RHS can be obtained by Edmond's greedy algorithm. \medskip

We now describe the dual solution $(\alpha(t),\tau(t))$ also on a continuous way. Indeed, the distribution $\tau(t)$ is precisely the product distribution induced by $y$.
Namely,
\[
\tau(t)_S := \prod_{j\in S} y(t)_j \cdot \prod_{j\notin S} (1 - y(t)_j)
\]
We hide the dependence on $t$ for brevity's sake. So note that $\Exp_{S\sim \tau} [f(S)]$ is {\em precisely} $F(y)$.

Next, for each element $j\in V$, define $\beta_j := \Exp_{S \sim \tau}f_S(j)$. Further, define 
$P(y) := \{p_1,p_2,\dots,p_k\}$ to be the ordered maximum weight independent set in $\calM$ where the weight of every element $j$ is $\beta_j$ picked by Edmond's greedy algorithm.
As in the previous section, define $T_i := \spann(\{p_1, \ldots, p_i\})$, and consider the $\alpha = \alpha(t)$ values
\[\alpha_{T_k} = \max_{j \in T_k}\beta_j~~~~\textrm{and}~~~~
\alpha_{T_i} = \left(\max_{j \in T_i \setminus T_{i-1}}\beta_j\right)-\sum_{\ell = i+1}^k\alpha_{T_\ell},~~\forall k > i \geq 1\]
\noindent
As in \eqref{eq:goopy}, the above $(\alpha, \tau)$ forms a feasible dual solution. As in \eqref{eq:bagha}, we have $\sum_{T\subseteq V} r_M(T)\alpha_T = \sum_{i=1}^k \left(\max_{j \in T_i \setminus T_{i-1}}\beta_j\right)$.
%
\begin{lemma}\label{lem:grad-alpha}
	$\sum_{T \subseteq V} r_M(T)\alpha_{T} \leq v^* \boldsymbol{\cdot} \grad F(y)$
\end{lemma}
\begin{proof}
	Note that $P(y) \in \mathcal{I}$ and so $\mathbbm{1}_{P(y)} \in \mathcal{P}_M$ and thus $v^* \boldsymbol{\cdot} \grad F(y) \geq  \mathbbm{1}_{P(t)} \boldsymbol{\cdot} \grad F(y)$. Additionally, note that (by definition of $F$) that $\forall p_i \in P(y), \mathbbm{1}_{\{p_i\}} \boldsymbol{\cdot} \grad F(y) = \Exp_{S \sim \tau}f_S(p_i) = \beta_{p_i}$, yielding
\arxiv{
	\begin{alignat}{2}
		\mathbbm{1}_{P(t)} \boldsymbol{\cdot} \grad F(y) ~=~ \sum_{p_i \in P(y)}\beta_{p_i} ~\underbrace{=}_{\text{Edmond's greedy algorithm}}~ \sum_{i=1}^k\max_{j \in I_i \setminus I_{i-1}}\beta_j ~=~ \sum_{T \subseteq V} r_M(T)\alpha_{T}  \notag\arxiv{\qedhere}
	\end{alignat}
}
\ipco{
$\mathbbm{1}_{P(t)} \boldsymbol{\cdot} \grad F(y) =~ \sum_{p_i \in P(y)}\beta_{p_i}$ which equals $\sum_{i=1}^k\max_{j \in I_i \setminus I_{i-1}}\beta_j$ by Edmonds greedy algorithm.
As noted above, this sums to $\sum_{T \subseteq V} r_M(T)\alpha_{T}$
}
\end{proof}
A consequence of the above lemma is that the dual objective of $(\alpha(t),\tau(t))$ satisfies
\arxiv{
\[
\dual(t) = \sum_{T \subseteq V} r_M(T)\alpha(t)_{T} + \Exp_{S \sim \tau(t)} [f(S)] \leq F(y(t)) + v^* \boldsymbol{\cdot} \grad F(y(t))~~\underbrace{=}_{\text{continuous greedy}} F(y(t)) + \frac{dF(y(t))}{dt}
\]
}
\ipco{
$\dual(t) = \sum_{T \subseteq V} r_M(T)\alpha(t)_{T} + \Exp_{S \sim \tau(t)} [f(S)] \leq F(y(t)) + v^* \boldsymbol{\cdot} \grad F(y(t)) = F(y(t)) + \frac{dF(y(t))}{dt}$
where the last equality follows by the continuous greedy algorithm's description.
}
Readers familiar with the continuous greedy analysis would now notice the familiar differential equation whose solution leads to the following lemma.
\begin{lemma} \label{clm:minfit}
	$\min_{t \in [0,1]}\dual_{y(t)} \leq (1-\frac{1}{e})F(y(1))$
\end{lemma}
Note that in the LHS we {\em do not} have $\dual(1)$; the reason is that the dual objective value is not necessarily monotonically decreasing. Nevertheless, the ``best'' dual seen so far provides
the dual certificate for the $(1-1/e)$-approximation (and could potentially be better in certain instances). 

We make a final remark to end the section. The dual $\tau(t)$ is described {\em implicitly} using the distribution $y(t)$ thus we {\em do not} bypass the randomization required to compute the value of the dual: the only way we know is to sample and estimate using the techniques described in \cite{CaliCPV11}. We leave the question of designing an optimal primal-dual algorithm for  \mmsm~as an open question.

%
%
}
\section{Conclusion}
In this paper we consider the classic and important problem of maximizing a monotone submodular function subject to a cardinality constraint from the point of view of achieving 
a principled and efficient approach to obtain {\em instance-wise upper bounds} which have provable optimal worst-case guarantees and also have good empirical performance. 
We do so via a primal-dual approach on an exponentially sized linear program of Nemhauser and Wolsey~\cite{NemhW81}. In particular, we devised a novel algorithm which returns a set $S$ and a $\dual$~upper bound with an optimal $\eta_k$ approximation guarantee, and also comparable empirical performance with respect to state-of-the-art. We also consider the generalization of
the LP to matroid constraints and give dual-fitting analyses for two common \mmsm~algorithms -- traditional greedy and continuous greedy -- with both dual-fitting upper bounds achieving their respective algorithm's optimal approximation guarantee. While this paper does not provide a novel algorithm for the \mmsm~problem, the duals do provide instance-wise upper bounds which could be potentially better, and more importantly, can provide new insight for future investigations towards an optimal deterministic algorithm for \mmsm.
From a technical standpoint, we give a primal-dual approximation algorithm for a {\em maximization} problem, adding to the small list of such examples. We believe that more research is warranted on the primal-dual schema for maximization problems and believe our work may spur interest. In particular, it is interesting to approach \mmsm, or even special cases of it like the generalized assignment problem, with this lens.



\arxiv{\subsection*{Acknowledgements}

The authors thank Ankita Sarkar who spurred this work by asking whether a primal-dual algorithm exists for the maximum coverage problem. We also thank Jan Vondr\'ak for enlightening email exchanges,
pointing us to relevant literature and suggesting that an explicit dual-fitting for continuous greedy may be possible. We thank anonymous reviewers for multiple suggestions and pointers to references which have improved the presentation of this material. During this work, DC was partially supported by NSF grant CCF-2041920 and LC was partially supported by an award from the Lovelace Research Program at Dartmouth College.
}

\bibliographystyle{abbrv}
\bibliography{refs}
\appendix

\section{Notes on the Primal LP}\label{sec:notes-on-primal}

\paragraph{Randomized Rounding.}
Given a feasible solution $(x,\lambda)$ to \eqref{eq:p1}, we show a simple randomized algorithm which returns a subset $R$ with $\Exp[f(R)] \geq (1-\frac{1}{e})\cdot \lp$.
The algorithm begins with $R = \emptyset$, and proceeds in $k$ independent rounds where in each round a single element $j\in V$ is chosen to be in $R$ with probability $\frac{x_j}{k}$.
Note, the same element could be picked multiple times into $R$ (which is a waste).
Let $R_t$ denote the set $R$ at the beginning of round $t$; so $R_1 = \emptyset$ and $R = R_{k+1}$.
Note that
$\Exp[f(R_{t+1}) - f(R_t)~|~ R_t] = \sum_{j\in V} \Pr[j~\textrm{picked at round $t$}]\cdot \left(f(R_t \cup j) - f(R_t)\right)$ and this equals $\frac{1}{k}\sum_{j\in V} f_{R_t}(j)x_j \geq \frac{\lambda - f(R_t)}{k}$.
Taking expectations again and rearranging
\[
\Exp[f(R_{t+1})] \geq \left(1 -\frac{1}{k}\right)\Exp[f(R_t)] + \frac{\lambda}{k}
\]
\Cref{clm:ana} now proves
$\Exp[f(R)] = \Exp[f(R_{k+1})] \geq \eta_k \cdot \lambda$.

\paragraph{Infeasibility of Solving \eqref{eq:p1}.}
The above LP has exponentially many constraints. Can it be solved? A sufficient condition is separability: given $(x_j, j\in V; \lambda)$, we want to know if 
\begin{equation}\label{eq:sep}
	\min_S \left(f(S)  + \sum_{j\in V} x_j f_S(j) \right) \geq \lambda
\end{equation}
or not. The right hand side is precisely the function $f^*(x)$ defined by Calinescu~\etal~\cite{CaliCPV07}.
We now show that evaluating $f^*(x)$ can be require exponentially many queries to $f$; this complements the NP-hardness result mentioned in~\cite{CaliCPV07}.
Therefore, if we can solve the minimization problem in the LHS, we can solve the above LP. 

In particular, we will show the hardness for $x$-vector with $x_1 = x_2 = 1$ and all the rest are $0$.
That is, we cannot solve the following optimization problem
\[
\min_{S\subseteq V}~~ \Big(f(S\cup 1) + f(S\cup 2) - f(S) \Big)
\]
with polynomially many calls to an oracle of a monotone submodular function. In the remainder of the section, we establish this fact. \medskip

\noindent
We think of $V = \{0,1\}^{n+2}$ where the first two elements are $1$ and $2$. We use bit-vectors and subset notation interchangably.
We first define a monotone submodular function on $\{0,1\}^n$ as follows. First, fix a subset $R$ of size $k$, where $k$ is any parameter and for the lower bound, $k=n/2$ works best.

\begin{equation}\label{eq:pav-m}
	r_(S) = \begin{cases}
		|S| & \textrm{if $|S| \leq k-1$} \\
		k  & \textrm{if $|S| = k$ and $S\neq R$} \\
		k-1 & \textrm{if $S = R$} \\
		k & \textrm{if $|S| > k$}
	\end{cases}
\end{equation}

\begin{claim}\label{clm:pavms}
	$r$ is a monotone submodular function.
\end{claim}
\begin{proof}
	Indeed, the above is an example of a rank function of a paving matroid. We give a self-contained proof below.
	Monotonicity should be clear, and indeed, $r_M(S) \leq |S|$ also should be clear from the definition.
	For any two sets $A$ and $B$, if both $|A|$ and $|B|$ are $< k$ or $=k$ but neither is $R$, then $r(A) + r(B) = |A| + |B|$ which equals $|A\cup B| + |A\cap B| \geq r(A\cup B) + r(A\cap B)$.
	If one of them, say $B$, is $R$ and the other has $|A| \leq k$, then $r(A) + r(B) = |A| + k - 1$, but $|A\cup B|$ is surely $> k$ then (since $A$ contains items not in $R$), and so $r(A\cup B) \leq |A\cup B| - 1$. Thus, again, $r(A) + r(B) = |A| + |B| - 1 = (|A\cup B| - 1) + |A\cap B| \geq r(A\cup B) + r(A\cap B)$.
	If any one of them, say $A$, has $|A| > k$, then $r(A) + r(B) = k + r(B) \geq k + r(A\cap B)$ by monotonicity, and this is $= r(A\cup B) + r(A\cap B)$ since $|A\cup B| > k$ as well.
\end{proof}

Next define the function on $\{0,1\}^{n+2}$ as follows. We use $00S$ to denote $S$, $01S$ to denote $S\cup 2$, $10S$ to denote $S\cup 1$, and $11S$ to denote $S\cup 1 \cup 2$.
\[
f(00S)  = r(S) + \min(|S|, k) ~~~ f(01S) = f(10S) = f(11S) = r(S) + k
\]

\begin{claim}
	$f$ is a monotone submodular function.
\end{claim}
\begin{proof}(Sketch)
	Monotonicity across dimensions $1$ and $2$ follow from the definition since $r(S) + \min(|S|,k)$ becomes $r(S) + k$.
	Across other dimensions it follows from~\Cref{clm:pavms}.
	Think of $f$ as defined on four hypercubes defined by the first two coordinates. We need to check for submodularity on ``squares''.
	If all endpoints of these squares lie in the same sub-hypercube, then submodularity follows since $r(S) + k$ and $r(S) + \min(|S|,k)$ are
	both submodular since it's a sum of two submodular functions.
		If all endpoints are in different squares then
	\[
	f(10S) + f(01S) - f(11S) - f(00S) = k - \min(|S|,k) \geq 0
	\]
	If the endpoints are $00S$, $00(S+i)$, $10S$ and $10(S+i)$, then
	\[
	f(10S) + f(00(S+i)) - f(00S) - f(10(S+i)) = \min(|S+i|,k) - \min(|S|,k) \geq 0
	\]
	If the endpoints are $01S$, $01(S+i)$, $11S$ and $11(S+i)$, then
	\[
	f(01(S+i)) + f(11S) - f(01S) - f(11(S+i)) = 0
	\]
\end{proof}

Finally, note that for any set $S\subseteq \{3,4,\ldots, n+2\}$,
\[
h(S):=	f(S\cup 1) + f(S\cup 2) - f(S) = f(01S) + f(10S) - f(S) = r(S) + k - \min(|S|,k)
\]
When $S = R$, we get $h(R) = k-1$. When $|S| = k$ but $S\neq R$, $h(S) = k$. When $|S| < k$, $h(S) = k$. When $|S| > k$, then $h(S) = k$.
When $S$ contains $1$ or $2$ (or both), then let $T := S\cap \{3,4,\ldots, n+2\}$.
Note that $h(S) = r(T) + k \geq k$.
In short, the minimum value of $h(S)$ is obtained by this set $R$. If you knew the size of $|R| = k$, then note that any query to $r(S)$ doesn't give {\em any} information about $R$.
Thus, one needs to make exponentially queries to find $R$. 
\section{Further Differences between the Primal-Dual and Greedy} \label{sec:greedy-comp}
In~\Cref{subsec:illexample}, we illustrated via an example how the primal-dual and greedy algorithm differ. 
Indeed, we show that in the standard ``bad'' example for the greedy algorithm for the max-coverage problem, the primal-dual algorithm in fact does well (although a slight jiggling of the instance leads to bad behaviour for both algorithms). 

In fact, if one recalls the example where the greedy algorithm actually achieves the $\eta_k$-factor (showing that the analysis is tight), in that instance the optimal solution is indeed disjoint ``large'' sets. To recall, the set system is as follows: there are $k$ ``good'' sets and $k$ ``bad'' sets. Each good set contains $x$ elements and intersects no other good set. We now assign an arbitrary ordering to the ``bad'' sets to distinguish them from one another. Each ``bad'' set intersects no other ``bad'' sets, but evenly intersects every ``good'' set. Additionally, the first ``bad'' set, $B_1$, contains $x$ elements, while for every other set, the $ith$ ``bad'' set, $B_i$ contains $B_i - \frac{B_i}{k}$ elements. An example of this set system where $k=2, x=8$ and is shown below:
\begin{figure}[H]
	\centering
	\includegraphics[width=0.6\textwidth]{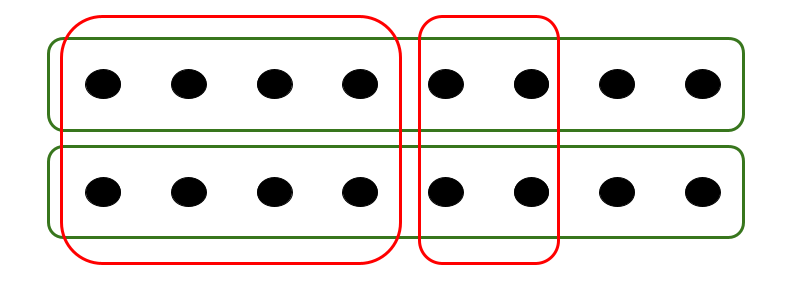}
	\caption{set system outlining the structure for the greedy worst-case example with ``good'' sets in green and ``bad'' sets in red.}
	\label{fig:greedy-worst}
\end{figure}
The optimal solution here is to pick all the ``good'' sets leading to all $kx$ elements being covered. Just as in the previous example, if you you allow one to adversarially make greedy's choices among equal contenders, you can force greedy to choose a ``bad'' set for every pick, leading to a total of $x\cdot k \left(1-\frac{1}{k}\right)^{k}$ elements being covered. Note that this quantity goes to $xk(1-\frac{1}{e})$ as $k \rightarrow \infty$. For the primal-dual algorithm, even if we adversarially make the algorithm pick the first ``bad'' set, $B_1$, the algorithm will still be unable to make dual progress as $\dot{\beta_{j^*}} = F_{j^*} - \beta{j^*} = (x - \frac{x}{k})-x = -\frac{x}{k}$ and $\dot{\gamma} = F-\gamma = x-0=x \implies k\dot{\beta_{j^*}}+\dot{\gamma} = 0$. This will cause the primal-dual algorithm to immediately pick another tight set, which must be one of the ``good'' sets. We then note that this will only cause $F$ to increase, while the marginal utility of all the other ``good'' sets will stay the same, continually disabling dual progress and causing the primal-dual algorithm to pick another tight set. Thus, even when forced to pick a ``bad'' set first, the primal-dual algorithm will pick $k-1$ ``good'' sets and only 1 ``bad'' set, covering a total of $(k-1)x + \frac{x}{k}$ elements. This, therefore, shows an example where the primal-dual algorithm is in fact 
$\approx e/e-1$ times better than the greedy algorithm. The reader, however, will note that even in the above example, if we apply a small fudge factor of $-\epsilon$ to $F_{j}$ for all the ``good'' sets the dual will now be able to make progress and this example itself turns into a worst case example for the primal-dual algorithm. 

In conclusion, the primal-dual algorithm is a different algorithm than the greedy algorithm and is incomparable if compared instance-by-instance (even for the max $k$-coverage problem). 
A better understanding of this algorithm may lead to a hybrid algorithm which empirically outperforms both these algorithms, and at the same time give a feasible dual which certifies the quality instance-by-instance. 
\section{On the Balkanski-Qian-Singer Upper Bound}\label{sec:bqs-dualfit}
In~\cite{BalkQS21}, Balkanski, Qian, and Singer give an interesting upper bound (called BQS) on the optimum value of \msm. Empirically, this upper bound has been observed by both their and our work to have the lowest (that is, best) value so far. As mentioned in the main body and reported in~\Cref{sec:experiment}, our primal-dual upper bound is comparable: certain instances we are better and certain instances they are, and none is ever far from the other. This section is motivated by the following seemed a natural question: does the BQS upper bound correspond to a feasible dual solution for the Nemhauser-Wolsey~\cite{NemhW81} LP? The answer is {\em no}; there exist ``integrality gap'' instances for the Nemhauser-Wolsey LP where the BQS upper bound equals the true optimum value while the LP value is strictly larger. Below, we describe this example; before we do so, we describe the BQS upper bound and give another self-contained proof of the fact that it is indeed an upper bound. 

\paragraph{The BQS Upper Bound.} We describe what is called ``Method 3'' in their paper\footnote{we refer to the implementation given in Appendix A of the arXiv version: \href{https://arxiv.org/pdf/2102.11911.pdf}{https://arxiv.org/pdf/2102.11911.pdf}}. They later obtain a better upper bound called ``Method 4'' which uses the greedy algorithm's output.
Method 3, however, is stand-alone, and we restrict our attention to this and it suffices for the example we show to differentiate from the Nemhauser-Wolsey LP. 

To begin with, let us rename the items of $V$ as $\{a_1, a_2, \ldots, a_n\}$ in decreasing order of $f({\{j\}})$. So, $f(\{a_1\}) \geq f(\{a_2\}) \geq \cdots \geq f(\{a_n\})$.
Next, for any $1\leq i\leq n$, define $A_i := \{a_1, \ldots, a_i\}$ be the first $i$ elements in this order which we unimaginatively call the $A$-order. Now we are ready to describe their algorithm.
\begin{center}
	\fbox{%
		\begin{varwidth}{\dimexpr\linewidth-2\fboxsep-2\fboxrule\relax}
				\begin{algorithmic}[1]
					\Procedure{BQS Method 3}{$f,V,k$}: 
					\For{$j=1$ to $k$}
					\State $i_j \gets \min \{i ~:~ f(A_i) - \sum_{\ell=1}^{j-1}v_\ell \geq f(a_i)\}$ \label{bqs:i}
					\If{$f(A_{i_j-1})-\sum_{\ell=1}^{j-1}v_\ell \geq f(a_{i_j})$} \label{bqs:if}
					\State $v_j \gets f(A_{i_j-1})-\sum_{\ell=1}^{j-1}v_\ell$ \label{bqs:vjif}
					\Else
					\State $v_j \gets f(a_{i_j})$ \label{bqs:else}
					\EndIf
					\EndFor
					\State \Return $\sum_{j=1}^kv_j$
					\EndProcedure
				\end{algorithmic}					
		\end{varwidth}%
	}
\end{center}
It is not immediately clear why $\sum_{j=1}^k v_j$ would be an upper bound on $\opt$.
Indeed, in their paper~\cite{BalkQS21} nicely describe the method which led to this upper bound (it goes via considering the minimum submodular set cover problem a la Wolsey~\cite{Wolsey82}).
Below we give a direct proof of this fact which may be instructional; the reader wishing to see the gap example may skip ahead. In fact, our proof shows that the BQS upper bound also holds
for the maximization problem of {\em subadditive} functions; a function is subadditive if $f(A \cup B) \leq f(A) + f(B)$.
\begin{theorem}
	For any monotone, {\em subadditive} function $f$, we have $\max_{S:|S|=k} f(S) \leq \sum_{j=1}^k v_j$, the solution returned by {\sc BQS Method 3}.
\end{theorem}
\begin{proof}
Note that the algorithm also terminates with $k$ indices $i_1 < i_2 < \cdots < i_k$. 
The ``if-condition'' of Line~\ref{bqs:if} leads to the following observation.
\begin{observation}\label{obs:bqs-o}
	For all $1\leq j \leq k$, $v_j \geq f(a_{i_j})$.
	For all $1\leq j \leq k$, $\sum_{\ell=1}^j v_\ell \geq f(A_{i_j-1})$.
\end{observation}
\noindent
Fix any set $S$ with $|S|\leq k$. 
Consider the elements of $S$ in the $A$-order, that is, $S = \{a_{j_1}, a_{j_2}, \ldots, a_{j_k}\}$ with $j_1 < j_2 < \cdots < j_k$.
Let $r$ be the largest index in $[k]$ such that $i_r > j_r$; if no such $r$ exists, then let this be $0$.
Note that for all $\ell > r$ we have $i_\ell \leq j_\ell$ which implies $f(a_{j_\ell}) \leq f(a_{i_\ell})$.
Also note $i_r > j_r$ implies $j_r \leq i_r - 1$, and so $f(S \cap A_{i_r - 1}) \leq f(A_{i_r - 1}) \leq \sum_{\ell=1}^r v_\ell$.
Thus,
\[
f(S) \underbrace{\leq}_{\text{subadditivity}} f(S\cap A_{i_r - 1}) + \sum_{\ell=r+1}^k f(a_{j_\ell}) \leq f(A_{i_r - 1}) + \sum_{\ell=r+1}^k f(a_{i_\ell}) \le \sum_{j=1}^k v_j\arxiv{\qedhere}
\]	
\end{proof}
\noindent
Although the above upper bound holds even for subadditive functions, the ``Method 4'' of~\cite{BalkQS21} applies the above method not on $f$ directly but on $f_S$ defined as $f_S(T) := f(T\cup S) - f(S)$
where $S$ is the solution of the greedy algorithm. More precisely, the upper-bound of Method 4 is $f(S) + \sum_{j=1}^k v^{(S)}_j$ where $v^{(S)}_1, \ldots, v^{(S)}_k$ is the solution returned by Method 3 run on the function $f_S$.
Note that this ``contraction'' preserves submodularity but not subadditivity, and so if $f$ were a subadditive function, then Method 4 may not return a valid upper bound (but would if $f$ were submodular). Balkanski~\etal also show that there exist submodular functions for which the upper bound of Method 3 can be very bad; but Method 4's output is always within factor $2$ of the optimum.
This argument is similar to the dual-fitting argument we mention in~\Cref{sec:2} how putting $\tau$-mass on $S$ gives a dual solution of at most $2f(S)$.

\paragraph{An Integrality Gap Example.} 
We now describe a simple example which is an integrality gap example for the Nemhauser-Wolsey LP, that is, the LP value is strictly larger than the integer optimum, but the BQS value equals the integer optimum. The universe has three elements $\{a_1, a_2, a_3\}$ and $k=2$. The function is defined as $f(\emptyset) = 0$, $f(a_1) = f(a_2) = 3$ and $f(a_3) = 1$. 
The function value on any set of two elements is $4$, while $f(\{a_1, a_2, a_3\}) = 5$. Clearly, the optimum value for \msm~is $4$.
The BQS upper bound is also the same with $v_1 = f(a_1) = 3$ and $v_2 = f(a_3) = 1$ giving $v_1 + v_2 = 4$.

The Nemhauser-Wolser LP~\eqref{eq:p1} has an integrality gap. To see this, consider $x_1 = x_2 = x_3 = 2/3$ and $\lambda = 4\frac{1}{3}$.
One can now check that $f(S) + \sum_{j} f_S(j)x_j$ is minimized for the sets $\{a_1\}$ and $\{a_2\}$, and for both the value is $4\frac{1}{3}$.

\ipco{\subsection*{Acknowledgements}

The authors thank Ankita Sarkar who spurred this work by asking whether a primal-dual algorithm exists for the maximum coverage problem. We also thank Jan Vondr\'ak for enlightening email exchanges,
pointing us to relevant literature and suggesting that an explicit dual-fitting for continuous greedy may be possible. We thank anonymous reviewers for multiple suggestions and pointers to references which have improved the presentation of this material. During this work, DC was partially supported by NSF grant CCF-2041920 and LC was partially supported by an award from the Lovelace Research Program at Dartmouth College.
}
\end{document}